\documentclass[onecolumn,superscriptaddress,notitlepage]{revtex4-1}

\usepackage{dcolumn}% Align table columns on decimal point
\usepackage{bm}% bold math

\usepackage{amsmath}
\usepackage{graphicx}
\usepackage{amsbsy,amsfonts}
\usepackage{amsthm}
\usepackage{natbib}
\usepackage{hyperref}

\usepackage{mathrsfs}

\newcommand{\ket}[1]{\left|#1\right\rangle}

\newcommand{\alphavar}{\mu}
\newcommand{\Lambdavar}{\Lambda}
\newcommand{\Noracle}{{\textit{QueryCost}}}
\newcommand{\Costvar}{C}
\newcommand{\Avar}{A}
\newcommand{\sparseham}{{H}}
\newcommand{\upp}{\Lambda_{\rm max}}
\newcommand{\nn}{\nonumber \\}
\newcommand{\fyfuncnoarg}{\textbf{Col}}
\newcommand{\fyfuncone}[4]{\fyfuncnoarg(#1,#2,#3,#4)}
\newcommand{\BBfuncnoarg}{\textbf{MatrixVal}}
\newcommand{\fynoarg}{\textbf{Q}\fyfuncnoarg}
\newcommand{\BBnoarg}{\textbf{Q}\BBfuncnoarg}
\newcommand{\fyone}[3]{\textbf{Q}\fyfuncnoarg(#1,#2,#3)}
\newcommand{\sparse}[1]{$#1$-sparse}
\newcommand{\kz}{k_{0}}
\newcommand{\dt}{\Delta t}
\newcommand{\initt}{t_0}
\newcommand{\Th}[1]{$#1^{\text{th}}$}
\newcommand{\numQubitsInH}{n_H}
\newcommand{\numBitsInT}{n_t}
\newcommand{\Upsilonfcn}{\Upsilon}
\newcommand{\interval}{\Upsilon}
\newcommand{\scalconst}{K}
\newcommand{\upb}{Y}

\def\squareforqed{\hbox{\rlap{$\sqcap$}$\sqcup$}}
\def\qed{\ifmmode\squareforqed\else{\unskip\nobreak\hfil
\penalty50\hskip1em\null\nobreak\hfil\squareforqed
\parfillskip=0pt\finalhyphendemerits=0\endgraf}\fi}

\newtheorem{theorem}{Theorem}
\newtheorem{lemma}[theorem]{Lemma}
\newtheorem{definition}{Definition}
\newtheorem{corollary}[theorem]{Corollary}

\newenvironment{proofof}[1]{\begin{trivlist}\item[]{\flushleft\it
Proof of~#1.}}
{\qed\end{trivlist}}

\begin{document}
\title{Simulating Quantum Dynamics On A Quantum Computer}
\author{Nathan Wiebe}
\affiliation{Institute for Quantum Information Science, University of Calgary, Alberta T2N 1N4, Canada}
\affiliation{Institute for Quantum Computing, University of Waterloo, Ontario N2L 3G1, Canada}
\author{Dominic W. Berry}
\affiliation{Institute for Quantum Computing, University of Waterloo, Ontario N2L 3G1, Canada}
\author{Peter H\o yer}
\affiliation{Institute for Quantum Information Science, University of Calgary, Alberta T2N 1N4, Canada}
\affiliation{Department of Computer Science, University of Calgary, Alberta T2N 1N4, Canada}
\author{Barry C. Sanders}
\affiliation{Institute for Quantum Information Science, University of Calgary, Alberta T2N 1N4, Canada}
\affiliation{Department of Physics \& Astronomy, University of Calgary, Alberta T2N 1N4, Canada}
\begin{abstract}
We present efficient quantum algorithms for simulating time-dependent Hamiltonian evolution of general input states using an oracular model of a quantum computer.  Our algorithms use either constant or adaptively chosen time steps and are significant
because they are the first to have time-complexities that are comparable to the best known methods for simulating time-independent
Hamiltonian evolution, given appropriate smoothness criteria on the Hamiltonian are satisfied.  We provide a thorough cost
analysis of these algorithms that considers discretizion errors in both the time and the representation of the Hamiltonian.
In addition, we provide the first upper bounds for the error in Lie-Trotter-Suzuki approximations to unitary evolution operators, that use adaptively chosen time steps.
\end{abstract}
\maketitle

\section{Introduction}\label{sec:intro}
The original motivation for quantum computers stemmed from Feynman's famous conjecture that quantum computers could efficiently simulate quantum physical systems \cite{Fey82}, whereas there is no known way to achieve this with classical computers.
This conjecture has spurred the construction of a number of quantum algorithms to efficiently simulate quantum systems under a Hamiltonian \cite{Llo96,Wie96,Zal98,BT98,AT03,Chi04,BACS07,CK10}.
However, these algorithms are primarily for time-independent Hamiltonians.
A simple extension to time-dependent Hamiltonians yields complexity scaling quadratically with the simulation time \cite{RC03}, a significant performance reduction over the near-linear scaling for the time-independent case \cite{BACS07}.
These issues can be resolved by generalizing the Lie--Trotter--Suzuki product formul\ae~to apply in the time-dependent case.
Such formul\ae~have already been developed \cite{Suz93,WBHS10}, but have not yet been applied  to quantum simulation algorithms.
Here we explicitly show how these formul\ae\ can be used in quantum algorithms to simulate time-dependent Hamiltonians with complexity near-linear in the simulation time.
We provide a number of improvements to further improve the efficiency, and a carefully accounting of the computational resources used in the simulation.

Lloyd was the first to propose an explicit quantum algorithm for simulating Hamiltonian evolution~\cite{Llo96}.
This algorithm is for systems that are composed of subsystems of limited dimension, with a time-independent Hamiltonian consisting of a sum of interaction terms.
The algorithm uses the Trotter formula to express the time evolution operator as a sequence of exponentials of these interaction Hamiltonians, which may be simulated efficiently.
As a result, the complexity of the algorithm scales as $O(\|H\|\dt)^2$, where $\dt$ is the evolution time, and $\|H\|$ is spectral norm of the Hamiltonian.

Aharonov and Ta-Shma~\cite{AT03} and Childs~\cite{Chi04} extended these ideas to apply to Hamiltonians that are sparse, but have no evident tensor product structure.
They use graph coloring techniques to decompose the Hamiltonian into a sum of one-sparse Hamiltonians, and use the Trotter formula and the Strang splitting \cite{Strang}, respectively, to write the evolution operator as a sequence of one-sparse exponentials.
The resulting sequence of exponentials can then be performed by a quantum computer.
The use of higher-order splitting formula reduces the complexity of Child's algorithm to $O(\|H\|\dt)^{3/2}$, and it was conjectured that even higher-order formul\ae~may lead to near-linear scaling~\cite{Chi04}.

This hypothesis was verified by Berry, Ahokas, Cleve and Sanders (BACS)~\cite{BACS07}.
They used Lie--Trotter--Suzuki formul\ae~\cite{Suz90} to generate arbitrarily high-order product formula approximations to the time-evolution operator, and gave an improved method for decomposing the Hamiltonian.
The use of the Lie--Trotter--Suzuki formul\ae~reduced the cost of their algorithm, causing it to scale as $(\|H\|\dt)^{1+o(1)}$.
An alternative approach using quantum walks can yield scaling strictly linear in $\|H\|\dt$ \cite{Chi09,BC09}.

High-order Trotter-like approximations for ordered operator exponentials are needed to extend the results of BACS to apply to the simulation of time-dependent
Hamiltonian evolution.
Such integrators were originally proposed by Suzuki~\cite{Suz93}, using a time-displacement superoperator.
This method is made rigorous in Ref.\ \cite{WBHS10}, where sufficiency criteria for the applicability of the formul\ae, as well as bounds for the error, are provided.

Here we explicitly apply these integrators to provide an algorithm for simulation of sparse time-dependent Hamiltonians, and
find that its complexity depends on the norms of $H(t)$ and its derivatives.
We show how adaptive time steps may be employed such that the complexity depends on average values of these norms, rather than the maximum values.
This approach provides improved efficiency in situations where the norms have a sharp peak, or a finite number of discontinuities.
For situations with singularities, we show how efficiency may be improved by adapting the order of the integrators.

We also improve the performance by specifying that the oracles that encode the Hamiltonian encode their outputs in polar form.
Given this encoding, the one-sparse exponentials may be implemented via a simple circuit.
In addition, we improve simulation efficiency by expressing the Hamiltonian as a sum in different bases.
We quantify the performance of our scheme by considering the errors that occur in every step of the algorithm, including errors
that occur because of discretization of the times used by our quantum oracles.
We provide a unified presentation taking account of all these factors, as they interact in nontrivial ways.

\section{Our Approach}

In this section, we provide a less technical explanation of the results and how they are obtained.
Then we give the rigorous proofs in the following sections.
The objective in a quantum simulation algorithm is to simulate evolution under the Schr\"{o}dinger equation
\begin{equation}
\label{eq:Schroedinger}
	\frac{\partial}{\partial t} \ket{\psi(t)}=-iH(t)\ket{\psi(t)},
\end{equation}
where $H(t)$ is the time-dependent Hamiltonian.
That is, the initial state $\ket{\psi(\initt)}$ is encoded in the qubits of the quantum computer, and we wish to obtain a state in the quantum computer encoding an approximation of the final state $\ket{\psi(\initt+\dt)}$.
A quantum computer simulation algorithm achieves this by applying an (encoded) approximation of the time-ordered exponential
\begin{equation}
	 U(\initt+\dt,\initt)=\mathcal{T}\exp\left\{-i\int_{\initt}^{\initt+\dt}H(u)\mathrm{d}u\right\},
\end{equation}
to the initial state in the quantum computer such that $\ket{\psi(\initt+\dt)}=U(\initt+\dt,\initt)\ket{\psi(\initt)}$.
Given any $\epsilon>0$, our goal is to obtain an approximation of the final state that is within trace distance $\epsilon$ of the true state.
This can be achieved \cite{BACS07} if the approximation, $\tilde U(\initt+\dt,\initt)$, satisfies
\begin{equation}
\label{eq:normdifference}
	\left\|U(\initt+\dt,\initt)-\tilde U(\initt+\dt,\initt)\right\|\le \epsilon,
\end{equation}
with~$\|\cdot\|$ defined to be the two-norm.

\subsection{Constant-Sized Time Step Simulation}
Our primary objective in this paper is to demonstrate a quantum simulation algorithm for time-dependent Hamiltonian evolution that has a time-complexity that scales as $\dt^{1+o(1)}$.  The simplest approach to do so involves combining the sparse Hamiltonian decomposition scheme of BACS \cite{BACS07}, together with the higher-order integrators from Refs.\ \cite{Suz93,WBHS10}.
BACS show that a Hamiltonian with sparseness parameter (the maximum number of nonzero elements in any nonzero row or column) of $d$ may be decomposed into $6d^2$ one-sparse Hamiltonians.
Given that the state is encoded on $n$ qubits, there is an additional factor of $\log^*n$ to the number of queries required for the simulation.
Here $\log^*$ represents the iterated logarithm function, and increases extremely slowly with $n$.
This factor arises because the decomposition requires $O(\log^*n)$ queries to the oracle for the Hamiltonian to perform the decomposition.

The BACS decomposition technique may be used in the time-dependent case.
However, one complication is that the sparseness can, in the completely general case, depend on time.
That is, the nonzero elements at one time can be different to those at a different time.
This would mean that the decomposition depends on the time, which makes the use of Lie--Trotter--Suzuki formul\ae~problematic.
To avoid that problem, we consider every matrix element that \emph{ever} attains a nonzero value to be non-zero and use the BACS decomposition algorithm on those matrix elements.  This allows us to directly apply their decomposition result.

Reference \cite{WBHS10} gives a result for exponentials of a general operator $A(t)$.
The result for Hamiltonian evolution follows by taking $A(t)=iH(t)$.
Then, for $H(t)=\sum_{j=1}^m H_j(t)$, by using a Lie-Trotter-Suzuki product formula that is accurate to order $2k$ \cite{Suz93}, simulation error within $\epsilon$ may be achieved using a number of one-sparse exponentials that scales as
\begin{equation}
O\left( mk \left( \frac {25} 3 \right)^{k} (\Lambda \dt)^{1+1/2k}/\epsilon^{1/2k} \right).
\end{equation}
Here $\Lambda$ is an upper bound on the derivatives of the Hamiltonians such that
\begin{equation}
\Lambda \ge \left( \sum_{j=1}^m \| H_j^{(p)}(t)\| \right) ^{1/(p+1)},
\end{equation}
for $t\in[\initt,\initt+\dt]$ and $p\in\{0,\ldots,2k\}$ \cite{endnote}.
The notation with superscript $(p)$ denotes repeated derivatives.
An upper bound is used, rather than the exact maximum value, because the oracles that are used only give matrix elements of the Hamiltonian, not the norm.  This is significant because methods for computing the norm of a matrix are often inefficient.
However, it is often possible to place an upper bound on the norm, even if it is not possible to determine $\Lambda$ exactly.

We convert this result into a number of oracle queries by multiplying the number of exponentials by the number of oracle queries
that are needed to simulate a one-sparse operator exponential.  By doing so, we find that if the Hamiltonian is sufficiently smooth
 then the query complexity of the algorithm scales as $(\Lambda\dt)^{1+o(1)}$.

More generally, we also consider the case where $H(t)$ has discontinuous derivatives at a finite number of times.
Such discontinuities are problematic because if a Lie--Trotter--Suzuki formula is used to integrate across such a discontinuity, then error bounds proved in~\cite{WBHS10} may not apply.  In some cases this can be rectified by reducing the order of the integrator, but this strategy is not applicable if the Hamiltonian is not at least twice differentiable.
Instead, we choose the time intervals to omit these points of discontinuity.
In order to use this approach, it is necessary that the norm of the Hamiltonian is adequately bounded, because otherwise there could be significant evolution of the system very close to the point of discontinuity.
Given this restriction it is possible to perform the simulation with complexity that is essentially unchanged.
The full result, with the required conditions, is given in Corollary \ref{cor:piecewise}.

It is important to note that the performance of our constant step size algorithm scales with the largest possible value of the norms of $H_j$ and their derivatives.  For some Hamiltonians, these values may only be large for a small fraction of the simulated evolution and so the algorithm may be inefficient.  Using adaptive time steps, we can overcome this problem and demonstrate complexities that scale with the average values of the norms of $H_j$ and their derivatives, given additional restrictions on the Hamiltonian are met.  We discuss this approach below.

\subsection{Adaptive time steps}
\label{sec:metad}
The above scaling is the direct application of the results of Refs.\ \cite{BACS07} and \cite{WBHS10}.
We improve this scaling for problems where $H(t)$ is badly behaved.
For example, the matrix $H(t)$ may be rapidly changing at some times and may be slowly varying at others.  In such cases, using constant step size methods may be inefficient because overly conservative time steps will be taken during time intervals in which the Hamiltonian is comparatively easy to simulate.
We address this by introducing adaptive time steps.
When using adaptive time steps, instead of choosing each time step to be $\dt/r$, a sequence of times $\{t_p\}$ are chosen such that $\initt< t_1<\cdots<t_r = \initt+\dt$.
The size of the time intervals can be varied, as can the order of the product formula within each interval.

We choose the duration of the time steps using a time-dependent function, $\Upsilon(t)$, that provides similar information to $\Lambda$, but at a specific time.  We express this function as
\begin{equation}
\label{eq:upb}
\Upsilon(t) \ge\left(\sum_{j=1}^m \left\| H_j^{(p)}(t)\right\|\right)^{1/(p+1)},
\end{equation}
for $p\in\{0,\ldots,2k\}$.
Throughout this work we use the notation for the average value, $\overline\interval(t_{b},t_a)$, defined by
\begin{equation}
\overline\interval(t_{b},t_a) :=\frac 1{t_b-t_a} \int_{t_a}^{t_b}\Upsilonfcn(t)\mathrm{d}t.
\end{equation}
The goal is to choose the time steps adaptively such that the number of queries depends on the average value of $\Upsilon(t)$ over the interval, rather than its maximum value (previously denoted $\Lambda$).
The full result is given in Theorem \ref{thm:adaptiveResult2}.
The basis of the method is to choose $r$ time intervals to limit the error within each time interval to be no greater than $\epsilon/r$.

To choose these time intervals appropriately, we need to know what the maximum value of $\Upsilon(t)$ is over a given interval, because the maximum value of $\Upsilon(t)$ dictates the simulation error over a short time step.
Ideally one would want a method of choosing the duration of these steps that depends only on the value of $\Upsilon(t)$ at the beginning of the interval, in order to avoid needing to know the value of $\Upsilon(t)$ over the entire interval.
We achieve this by requiring that the derivative of $\Upsilon(t)$ is appropriately bounded.
A bound on the derivative is also necessary to obtain a result depending on the average of $\Upsilon(t)$.
This is because we can demonstrate that if the value of $\Upsilon(t)$ is approximately constant during each time step,
then the complexity of
the algorithm scales with the average value of $\Upsilon(t)$ over all time steps.  We then require that the derivative of $\Upsilon(t)$ is bounded in order to guarantee this approximate constancy in the limit of short time steps.

Even given the restriction on the derivative, it is unclear how to effectively choose the time intervals.
The problem is that the time intervals are chosen to ensure that the error in each interval is no greater than $\epsilon/r$, but then the number of intervals will depend on how the intervals were chosen.
To break this circular logic, we need a way of choosing an $r_g$, such that when we choose the time intervals to have error no greater than $\epsilon/r_g$, the total number of intervals is no larger than $r_g$.
The full technique is given in the proof of Theorem \ref{thm:adaptiveResult2}. The value of $r_g$ is chosen as in Eq.\ \eqref{eq:rdef}.
The duration of each time step can then be calculated using just information about $\Upsilon(t)$ at the beginning of the time step, via an increment inversely proportional to $\Upsilon(t_p)$, as given in Eq.\ \eqref{eq:recur}.

\subsection{Resource Analysis}
Our cost analysis of these simulation algorithms focuses on the number of queries that are made to a pair of quantum oracles that provide information about the locations and values of the nonzero matrix elements of the Hamiltonian.
We provide improvements to these oracles that enable us to improve the efficiency of the simulation, as well as to more precisely quantify the resource usage.
First, we require that the output of the oracle for the values of the matrix is encoded using a qubit string as a complex number in polar form.
The advantage of this is that the one-sparse Hamiltonian evolution can be applied using a simple circuit with qubit rotations proportional to the magnitude and phase of the matrix element.
This is shown in Sec.\ \ref{sec:oracles}, and the explicit circuit is given in Fig.\ \ref{fig:simcircuit}.
Furthermore, the qubit rotations may be performed independently for each qubit of precision yielded by the oracle.

To more precisely quantify the resource usage, we examine the number of qubits that the oracles need to provide, as well as the number of bits needed to represent the time.
The oracles that are traditionally used in quantum simulation algorithms yield many-qubit approximations to the matrix elements in a single query.
The fact that the qubits yielded by the oracle may be used independently motivates using oracles that output only one qubit per query.
Doing so further reduces the number of qubits needed to simulate a one-sparse Hamiltonian evolution, because qubits that are not currently being used need not be stored.

We find that the total number of qubits accessed for the positions of the nonzero matrix elements scales as $n\log^* n$ (see Lemma \ref{lem:oraclem}), due to the need to access each of the $n$ qubits for the position.
This yields a factor of $n$ increase in the apparent complexity over that if all qubits can be obtained in a single query.
The number of qubits accessed for the values of the matrix elements is independent of $n$, but it does depend on other simulation parameters such as the error tolerance, the evolution time, the sparseness of the Hamiltonian, the norm of the derivative of $H$ and $k$.
It depends on all of these quantities logarithmically, except for $k$.
The precise result for the number of oracle queries used is given in Lemma \ref{lem:adError}.

We also consider a new source of error that is unique to the simulation of time-dependent Hamiltonians: the discretization of the time.
Because the Hamiltonian varies with time, inaccuracy in the time will result in inaccuracy in the estimate of the Hamiltonian.
This is potentially problematic for Lie-Trotter-Suzuki product formul\ae, which rely upon precise times in order to obtain higher-order scaling for the error.
The higher-order scaling is not strictly obtained when the time is discretized, so it is necessary to show that the product formul\ae\ are not overly sensitive to error in the time so that we can use our discrete oracle in place of the continuous Hamiltonian.
Another source of difficulty is in simulating systems with discontinuities.
The problem is that the technique to avoid discontinuities requires choosing times arbitrarily close to, but on one side of, the discontinuity.
With discretization of the time, the rounding may yield times on the wrong side of the discontinuity.

In contrast to the output from the oracle, there is no need to use a coherent superposition of times, and the time may be regarded as a purely classical quantity at all stages in the calculation.
This means that it is less challenging to provide the time to high accuracy than it is to obtain output from the oracle to high accuracy.
Nonetheless, it is important for the reliability of the simulation that it is not unstable with inaccuracy in the time.
We find that the simulation is stable with the time precision.
The precision required depends on $k$, $d$, $\dt$ and the maximum norm of the derivative of the Hamiltonian.
It is logarithmic in all these quantities except for $k$; see Lemma \ref{lem:adError} for the full result.
The time precision also affects the result for simulating evolution with discontinuities in Corollary \ref{cor:piecewise}.
That result holds provided the time discretization is no greater that the time between discontinuities (condition 4).

\section{Background and definitions}
\label{sec:background}
Next we describe the technical background and definitions needed to understand our full results, which are given in the next section.
For generality, we consider a Hamiltonian that is not sparse in any known basis, but is the sum of Hamiltonians that are each sparse in their own canonical basis.
That is,
\begin{equation}
H(t)=\sum_{\alphavar=1}^M H_{\alphavar}'(t),
\end{equation}
where the set of operators $\{H_{\alphavar}':\mathbb{R}\mapsto \mathbb{C}^{N\times N}; \alphavar=1,\ldots,M\}$ is sparse when represented in their canonical bases.
We express $H_{\alphavar}'(t)=T_\alphavar^{\dagger} H_\alphavar(t) T_\alphavar$, where $H_\alphavar(t)$ is sparse in the computational basis, and $T_\alphavar$ is a basis transformation that maps the computational basis to the canonical basis of $H_\alphavar(t)$.

To avoid complications due to the nonzero elements changing as a function of time, we consider only those elements which are nonzero at any time.
We then take $d$ to be an upper bound on the number of elements in any row that are nonzero at any time in the interval of interest:
\begin{definition}
The set of operators $\{H_\alphavar\}$, where~$H_\alphavar:\mathbb{R}\mapsto\mathbb{C}^{N\times N}$,
is $d$-sparse on~$\mathcal{S}\subseteq \mathbb{R}$ if for each~$\alphavar$
there are at most $d$ matrix elements in each row of~$H_\alphavar(t)$ that attain a nonzero value for any $t\in\mathcal{S}$.
\end{definition}
Given a sparse Hamiltonian, BACS \cite{BACS07} provide a method to decompose the Hamiltonian into a sum of $6d^2$ one-sparse Hamiltonians, and express the evolution as a product of evolutions under each of these one-sparse Hamiltonians.
The definition of sparseness we use here is compatible with that of BACS, and therefore each Hamiltonian $H_\alphavar(t)$ may be expressed as a sum of one-sparse Hamiltonians.
That is,
\begin{equation}
H(t)=\sum_{\alphavar=1}^M\sum_{j=1}^{6d^2} T_\alphavar^\dagger H_{\alphavar,j}(t)T_\alphavar,\label{eq:backgrounddecomp}
\end{equation}
where each~$H_{\alphavar,j}(t)$ is one-sparse.

This decomposition is enabled by a quantum oracle that answers queries about the locations and values of the nonzero matrix
elements of the Hamiltonian~\cite{AT03,Chi04,BACS07,CK10}.
For the BACS decomposition, the matrix elements of each $H_{\alphavar,j}(t)$ in~\eqref{eq:backgrounddecomp} can be calculated using $O(\log^* n)$ queries to a quantum oracle for $H_\alphavar(t)$.
This oracle yields the locations and values of the nonzero matrix elements in a specified row of $H_\alphavar$ to arbitrary precision~\cite{BACS07}.
In this case the cost of computing the matrix elements to sufficient precision is concealed within the definition of the oracle.
To explicitly take account of the number of qubits that the oracle must provide, we separate it into two oracles that yield the positions and values of the nonzero matrix elements, and yield one qubit per query.
See Sec.\ \ref{subsec:oraclecalls} for the explicit form of the oracles.

To express the evolution under the Hamiltonian as a product of evolutions under the one-sparse Hamiltonians, BACS use the product formul\ae\ of Suzuki \cite{Suz90}.
Suzuki first proposed arbitrary-order product formul\ae\ for both time-independent and time-dependent cases~\cite{Suz90,Suz93}.
These product formul\ae\ are for operator exponentials, and are not limited to Hamiltonians.
Subsequent work by the present authors showed that Suzuki's approximation method may be less accurate than expected if the operator does not vary sufficiently smoothly with time~\cite{WBHS10}.
That work provided upper bounds for the approximation error, given that the operator does satisfy a smoothness requirement.

In this work, we use the result from Ref.\ \cite{WBHS10} upper bounding the approximation error, with the operators obtained via the decomposition method of BACS.
We therefore adopt the terminology ``$P$-smooth'' and ``$\Lambda$-$P$-smooth'' from Ref.\ \cite{WBHS10}.
These are formally stated below.
\begin{definition}\label{def:suzsmooth}
The set of operators $\{H_j:j=1,\dots, m\}$ is $P$-smooth on~$\mathcal{I}\subseteq\mathbb{R}$ if, for each~$H_j$, the quantity $\max_{p=0,\ldots,P}\left\|H_j^{(p)}(t)\right\|$ is finite on~$\mathcal{I}$.
\end{definition}
\begin{definition}\label{def:lambdasuzsmooth}
The set of operators $\{H_j:j=1,\dots, m\}$ is $\Lambda$-$P$-smooth on~$\mathcal{I}\subseteq \mathbb{R}$ if $\{H_j\}$ is $P$-smooth and
$\Lambda\geq\left(\sum_{j=1}^m\left\|H_j^{(p)}(t)\right\|\right)^{1/(p+1)}$,
for all $t\in\mathcal{I}$ and~$p\in\{0,1,\cdots,P\}$.
\end{definition}

The $P$-smooth requirement is needed in order to achieve an approximation of a given order, and the $\Lambda$-$P$-smooth requirement is needed to bound the error.
In this work we also consider adaptive time steps, and then it is the upper bound on the derivatives as a function of time that is important.
We therefore introduce the following definition.
\begin{definition}\label{def:adasmooth}
The set of operators $\{H_j:j=1,\dots, m\}$ is $\Upsilonfcn$-$P$-pointwise-smooth on the interval~$\mathcal{I}\subseteq \mathbb{R}$,
for $\Upsilonfcn : \mathbb{R}\mapsto\mathbb{R}$, if $\{H_j\}$ is $P$-smooth and
$\Upsilonfcn(t)\geq\left(\sum_{j=1}^m \|H_j^{(p)}(t)\|\right)^{1/(p+1)}$,
for all $t\in\mathcal{I}$ and~$p\in\{0,1,\cdots,P\}$.
\end{definition}
In addition we adopt the terminology that a set of Hamiltonians is $\Lambda$-$\infty$-smooth if the set is $\Lambda$-$2k$-smooth for every $k>0$.  Similarly, we say that a set of Hamiltonians is $\Upsilon$-$\infty$-pointwise-smooth if it is $\Upsilon$-$2k$-pointwise-smooth for every $k>0$.

Provided $\{H_j\}$ is $\Lambda$-$2k$-smooth, for the Hamiltonian $H(t)=\sum_{j=1}^m H_j(t)$ the evolution $U(\initt+\dt,\initt)$ may be approximated via the integrator $U_k$, which is given iteratively via \cite{WBHS10}
\begin{align}
\label{eq:suzuki2}
	U_1(\initt+\dt,\initt)
		&:= \prod_{j=1}^m\exp\left\{-iH_j(\initt+\dt/2)\dt/2\right\}\prod_{j=m}^1
		 \exp\left\{-iH_j(\initt+\dt/2)\dt/2\right\},\nonumber\\
	U_\ell(\initt+\dt,\initt)
		&:= U_{\ell-1}(\initt+\dt, \initt+(1-s_\ell)\dt)U_{\ell-1}(\initt+(1-s_\ell)\dt,\initt+(1-2s_\ell)\dt)\nonumber\\
		&\quad \times U_{\ell-1}(\initt+(1-2s_\ell)\dt,\initt+2s_\ell\dt)U_{\ell-1}(\initt+2s_\ell\dt,\initt+s_\ell\dt)U_{\ell-1}(\initt+s_\ell\dt,\initt),
\end{align}
with $	s_\ell=1/(4-4^{1/\left(2\ell-1\right)})$.
This formula is implied by Suzuki's work~\cite{Suz93}, but is stated explicitly in~\cite{WBHS10}.
The approximation error is $O((\Lambda\dt)^{2k+1})$, so this formula is appropriate for short time intervals.
For longer~$\dt$, the evolution time may be divided into~$r$ subintervals, resulting in the approximation
\begin{equation}
\label{eq:Ukr}
	U\left(\initt+\dt,\initt\right)
		\approx\prod_{\ell=1}^r U_k\left(\initt+\ell\frac{\dt}{r},\initt+(\ell-1)\frac{\dt}{r}\right).
\end{equation}
The value of $r$ is then chosen large enough such that the overall error is no greater than some allowable error, $\epsilon$.
This choice of $r$ scales as $O((\Lambda\dt)^{1+1/2k}/\epsilon^{1/2k})$.
More precisely, the error in the product formula will be no greater than $\epsilon$ if we take
\begin{equation}
\label{eq:rval2}
	r=\left\lceil2\epsilon^{-1/2k}\left( 2k(5/3)^{k-1}\Lambdavar\dt\right)^{1+1/2k}\right\rceil,
\end{equation}
provided that
\begin{equation}
\label{eq:WBHS10condition}
	\epsilon\le (9/10)(5/3)^k\Lambdavar\dt.
\end{equation}
This result is equivalent to Lemma 5 of~\cite{WBHS10}, after eliminating $Q_k$ by using the inequalities in Eq.~(A.3) of that paper.
The overall complexity of the simulation is then proportional to the value of $r$.

\section{Results}
\label{sec:results}
This section formally presents our main results.
The major result is an upper bound for the query complexity used to simulate time-dependent Hamiltonian evolution, using adaptive time steps and oracles that cost one query per yielded qubit.
In order to quantify the complexity, the primary goal is to bound the error.
In simulation schemes for time-dependent Hamiltonians there are three sources of error:
\begin{enumerate}
\item integrator error from using~$U_k$,
\item error due to using a finite-bit representation of the time, and
\item error due to using a discretized representation of~$H_\alphavar$.
\end{enumerate}
To guarantee that the total error in the simulation is $\epsilon$, we ensure that the latter two errors sum to at most half the total error tolerance.
The time steps are then chosen such that the contribution to the error from the integrators at most adds up to the remaining half.
The following lemma, proved in Appendix~\ref{sec:round-off}, yields upper bounds for the number of bits of precision for these quantities that are needed to ensure that the roundoff errors are at most $\epsilon/2$.

\begin{lemma}\label{lem:adError}
Let $\{H_\alphavar:\mathbb{R}\mapsto\mathbb{C}^{2^n\times2^n};\alphavar=1,\ldots,M\}$ be a set of time-dependent Hermitian operators that is $2$-smooth and~$d$-sparse on~$\mathcal{I}$, which we take to be the union of disjoint closed intervals  $\{\mathcal{I}_j\}$.
Furthermore define $\dt:=\sup_{t_i,t_f\in\mathcal{I}}(t_f-t_i)$.
The round-off error in simulating the product of the time-evolution operators that are generated by $H(t)=\sum_\alphavar T_\alphavar^{\dagger}H_\alphavar(t) T_\alphavar$ over each $\mathcal{I}_j$, using the integrator $U_k$ is $\epsilon/2$ if
\begin{enumerate}
\item the number of bits of precision used to represent the time, $\numBitsInT$, and the number of qubits of precision used to represent the matrix elements, $\numQubitsInH$, satisfy
    \begin{align}
	\numBitsInT&\ge \left\lceil\log_2\left(\frac{(\max_{t\in \mathcal{I},\alphavar}\|\partial_tH_{\alphavar}(t)\|)(32kMd^2)(5/3)^{k-1}\dt^2}{\epsilon}\right)\right\rceil, \nn
\label{eq:n'eq0}
\numQubitsInH&\ge 2\left\lceil\log_2\left(\frac{32k M d^2(5/3)^{k-1}\upp\dt}{\epsilon}\right)\right\rceil+6,
\end{align}
\item and for each~$j$,
the length of the subinterval~$\mathcal{I}_j$ obeys
\begin{equation}
	|\mathcal{I}_j|\ge \dt/2^{\numBitsInT}, \label{eq:notshort}
\end{equation}
\end{enumerate}
where~$\upp$ is an upper bound on~$\max_{t\in \mathcal{I},\alphavar} \|H_{\alphavar}(t)\|$.
\end{lemma}

Given these conditions on the number of bits of precision, we can then bound the number of queries to simulate the Hamiltonian using adaptive time steps as in the following theorem.

\begin{theorem}\label{thm:adaptiveResult}
If $\{H_\alphavar:\mathbb{R}\mapsto\mathbb{C}^{2^n\times 2^n};\alphavar=1,\ldots,M\}$ is a set of time-dependent Hermitian operators that is $\Upsilonfcn$-$2k$-pointwise-smooth, $d$-sparse on~$\mathcal{I}=[\initt,\initt+\dt]$, $\numBitsInT$,  and~$\numQubitsInH$ satisfy~\eqref{eq:n'eq0}, and there exists $\scalconst\in\mathbb{R}$ such that
$|\partial_t\Upsilonfcn(t)| \le \scalconst^2[\Upsilonfcn(t)]^2$ $\forall~t\in\mathcal{I}$, then for any $\epsilon\in (0,1]$, the evolution generated by $H(t)=\sum_\alphavar T_\alphavar^{\dagger}H_\alphavar(t) T_\alphavar$
can be simulated within error $\epsilon$, and with the number of queries (denoted~$\Noracle$) to our Hamiltonian oracles, satisfying
\begin{align}
\Noracle &\in O\left(\big[n\log^*n+\log(kM(5/3)^kd^2\upp\dt/\epsilon)\big]N_{\exp} \right)\label{eq:Noracleresult}, \\
N_{\exp} &\in O\left( Md^{2}k(25/3)^k\frac{\left[d^2\bar\interval(\initt+\dt,\initt)\dt\right]^{1+1/2k}}{\epsilon^{1/2k}}\right),\label{eq:NBTresult}
\end{align}
where~$\log^*$ is the iterated logarithm function and~$\upp$ is given in Lemma~\ref{lem:adError}.  The number of calls to~$\{T_\alphavar\}$, namely~$N_{\rm T}$, satisfies $N_{\rm T}\in O\left( N_{\exp}/(3d^2)\right)$.
\end{theorem}

This result is stated in asymptotic (big-O) notation wherein we take
\begin{equation}
\bar\interval(\initt+\dt,\initt)\dt,\ \upp\dt,\ \max_{t\in \mathcal{I},\alphavar}\|\partial_t \sparseham_{\alphavar}(t)\|\dt^2,\ M,\ d,\ n,\ \epsilon^{-1},\ k
\end{equation}
to be our asymptotic parameters.
Specifically, we say for two functions of these parameters, $f$ and~$g$, that~$f \in O(g)$ if there exists a constant $a>0$
such that~$|f|\le a |g|$ if \emph{all} of these parameters are sufficiently large.

If $H(t)$ is the sum of sufficiently smooth terms and bounded on~$t\in[\initt,\infty)$, then~$k$  can be chosen such that~\eqref{eq:Noracleresult} and~\eqref{eq:NBTresult} scale nearly linearly with~$\bar\interval(\initt+\dt,\initt)\dt$.
This observation is important because linear scaling is known to be a lower bound, and therefore our time scaling is nearly optimal~\cite{BACS07,CK09}.
This observation is stated formally in the following corollary.

\begin{corollary}\label{cor:linscale}
If, in addition to the requirements of Theorem~\ref{thm:adaptiveResult}, $\{H_\alphavar\}$ is $\Upsilonfcn$-$\infty$-pointwise-smooth and \sparse{d} on~$[\initt,\infty)$ then,
\begin{equation}
	\Noracle\in [n\log^*n+\numQubitsInH]Md^4\bar\interval(\initt+\dt,\initt)\dt\big(d^2\bar\interval(\initt+\dt,\initt)\dt/\epsilon\big)^{o(1)}. \label{eq:Noraclescale2}
\end{equation}
\end{corollary}

Before proceeding to show how to perform simulations with adaptive time steps, we first give the explicit scheme without adaptive time steps, and show how to take account of discontinuities in the Hamiltonian.

\section{Simulating Time-Dependent Hamiltonians}
\label{sec:timedep}

In this section, we show how to simulate time-dependent Hamiltonian evolution on a quantum computer.
In particular, we show that sparse time-dependent Hamiltonian evolution can be simulated efficiently provided that the Hamiltonian is at least piecewise twice-differentiable.
We also show that if $H(t)$ is sufficiently smooth,  then the query complexity of our simulation scheme is comparable to the cost of the BACS algorithm for
simulating time-independent Hamiltonian evolution.
First we give the explicit description of the oracles used, then we give the precise result for the complexity of the simulation in terms of these oracles.

\subsection{Oracle Calls}
\label{subsec:oraclecalls}

The time-dependent Hamiltonian over an $N$-dimensional Hilbert space~$\mathscr{H}$ can be represented by a matrix with
elements~$H_{xy}$, for~$x$ the row number and~$y$ the column number.
We consider a quantum oracle that can be queried to provide information about the locations and values of the nonzero matrix elements.
For additional generality, in this work we assume that the Hamiltonian is in the form
\begin{equation}
	H(t)=\sum_{\alphavar=1}^M T^\dagger_\alphavar H_\alphavar(t) T_\alphavar,\label{eq:hamtrans}
\end{equation}
where~$\{H_\alphavar\}$ is \sparse{d}.
This takes account of cases where the overall Hamiltonian is not sparse, but it may still be simulated efficiently using efficient basis transformations $T_\alphavar$ \cite{Zal98,BT98,RC03,Chi09}.
We therefore require oracles to give the locations and values of the nonzero elements of the matrix representation of each $H_\alphavar(t)$.

Our first oracle provides the column numbers of the nonzero matrix
elements in a given row of any $H_\alphavar$.  The function yields a requested bit
of a particular entry in a list of~$d$ column numbers.  This list contains the column number of every matrix element in a specified
row of~$H_\alphavar$ that attains a nonzero value.  This function is $\fyfuncnoarg$, where~$\fyfuncone{p}{i}{x}{\alphavar}$ yields the \Th{p} bit of the \Th{i}
potentially nonzero matrix element in row $x$ of~$H_{\alphavar}$.

Our second oracle provides a requested matrix element of~$H_\alphavar(t)$.
This function yields a requested bit of a binary encoding of a given matrix element evaluated at a specified time.  We denote this function as~$\BBfuncnoarg$, and define $\textbf{MatrixVal} ({p},{x},{y},{\alphavar},q)$ to yield
the \Th{p} bit of a binary encoding of the matrix element $\left[H_{\alphavar}(t_q)\right]_{xy}$.
To take account of discretization of the time, the time is specified by an integer $q$, which gives a time from a finite mesh $\{t_q\}$, with
\begin{equation}
	t_q = \initt+(q-1/2)\dt/2^{\numBitsInT},
\end{equation}
where $\numBitsInT$ is a positive integer and $[\initt,\initt+\dt]$ contains the simulation time interval.
We choose the matrix elements to be encoded in polar form, $(H_{\alphavar}(t))_{xy}=\rho(t)\exp(i\phi(t))$, where~$\rho$ and~$\phi$ are real numbers.
For convenience, we also assume that~$\rho(t)$ is encoded as
\begin{equation}
\label{eq:encode}
	\upp\left(\frac{\rho_1}{2}+\frac{\rho_2}{2^2}
		 +\cdots+\frac{\rho_{\numQubitsInH}}{2^{\numQubitsInH}}\right),
\end{equation}
where~$\rho_j$ refers to the \Th{j} bit of~$\rho$, and~$\upp$ is an upper bound for
$\max_{t\in \mathcal{I}}\max_{\mu=1\ldots M}\|H_{\alphavar}(t)\|_{\rm max}$.

The quantum oracles are unitary operations that give the values of these functions.
That is, for classical inputs~$p$, $i$ and~$\alphavar$ and the quantum input~$\ket{x}$,
\begin{equation}
	 \fyone{p}{i}{\alphavar}\ket{x}\ket{0}=\ket{x}\ket{\fyfuncone{p}{i}{x}{\alphavar}}.
\end{equation}
This differs from Ref.\ \cite{BACS07}, where $i$ was given as a quantum input.
In that work no superposition over the $i$ was needed, so it does not change the analysis to give it as a classical input.
Similarly, for classical inputs~$p$, $\alphavar$ and~$q$ and the quantum inputs~$\ket{x}$ and~$\ket{y}$,
\begin{equation}
	 \BBnoarg(p,\alphavar,q)\ket{x}\ket{y}\ket{0}=\ket{x}\ket{y}\ket{\textbf{MatrixVal} ({p},{x},{y},{\alphavar},q)}.\label{eq:BBdef}
\end{equation}

In the following section, we present asymptotic estimates of the query complexity for simulating time-dependent Hamiltonian evolution using a sequence of approximations $U_k$, which are implemented on a quantum computer equipped with oracles~$\BBnoarg$, $\fynoarg$ and~$\{T_\alphavar\}$.
We will quantify the number of calls to $\BBnoarg$ and $\fynoarg$ together as $\Noracle$, and quantify the number of calls to $\{T_\alphavar\}$ separately as $N_{\rm T}$.

\subsection{Constant Timestep Simulation Method}\label{subsec:constanttime step}

The simulation problem is as follows.
Given the oracles~$\BBnoarg$, $\fynoarg$ and the set of oracles~$\{T_\alphavar\}$, we wish to simulate the evolution generated by the Hamiltonian given in~\eqref{eq:hamtrans}, for $\{H_\alphavar\}$ $\Lambda$-$2k$-smooth and \sparse{d}.
In addition, the user is provided with an upper bound for the norm of each~$H_\alphavar(t)$ and a similar upper bound for the norms of its derivatives.
Our task is to provide an upper bound for the number of oracle queries that are needed to simulate the evolution generated by $H(t)$ within error $\epsilon$.
These upper bounds are given in the following lemma.

\begin{lemma}\label{lem:timedepsim}
If $\{H_\alphavar:\mathbb{R}\mapsto\mathbb{C}^{2^n\times 2^n};\alphavar=1,\ldots,M\}$  is a set of time-dependent Hermitian operators that is $\Lambda$-$2k$-smooth
and \sparse{d} on~$[\initt,\initt+\dt]$, then for any $\epsilon\in (0,1]$ the evolution generated by $H(t)=\sum_\alphavar T_\alphavar^{\dagger}H_\alphavar(t) T_\alphavar$ can be simulated with the error due to the integrator bounded by $\epsilon/2$, and a number of queries to $\fynoarg$ and $\BBnoarg$ satisfying
\begin{align}
\Noracle&\leq 12 CMd^2 5^{k-1} \left\lceil 24kd^2\Lambdavar\dt \left(\frac{5}{3} \right)^{k}\left(\frac{6d^2\Lambdavar\dt}{\tilde\epsilon/2}\right)^{1/2k}\right\rceil\label{eq:Noracle0},
\end{align}
where~$\tilde{\epsilon}:= \min\{\epsilon,18(5/3)^{k-1}d^2\Lambdavar\dt\}$, $\Costvar$ is the number of oracle calls needed to
 simulate a one-sparse Hamiltonian.
The number of queries to $\{T_\alphavar\}$, $N_{\rm T}$, is bounded above by $\Noracle/(3\Costvar d^2)$.
\end{lemma}

\begin{proof}
Our approach is to express $H(t)$ as a sum of~$6Md^2$ one-sparse Hamiltonians and use~\eqref{eq:Ukr} to approximate $U(\initt+\dt,\initt)$ with a sequence of these one-sparse operator exponentials.
The number of oracle calls in the simulation is then obtained by multiplying the number of one-sparse exponentials in our approximation by the number of oracle calls to simulate one-sparse Hamiltonian evolution.
That number is just given as $C$ here, and an upper bound will be placed on it in Lemma~\ref{lem:oraclem}.

Each $d$-sparse Hamiltonian $H_\alphavar$ may be decomposed into a sum of~$6d^2$ one-sparse Hamiltonians $H_{\alphavar,j}$ by using the
BACS decomposition scheme.
As each matrix element of~$H_\alphavar$ is uniquely assigned to a one-sparse Hamiltonian $H_{\alphavar,j}$ in the BACS decomposition~\cite{BACS07}, we have that~$\|H_{\alphavar,j}^{(p)}(t)\|\le \|H_\alphavar^{(p)}(t)\|$ for any non-negative integer $p\le 2k$.
Using Definition~\ref{def:lambdasuzsmooth}, if $\{H_\alphavar\}$ is $\Lambda$-$2k$-smooth, then the set of Hamiltonians $\{H_{\alphavar,j}\}$ is $6d^2\Lambda$-$2k$-smooth.
Using Eq.~\eqref{eq:rval2} and the fact that~$N_{\exp}=2m5^{k-1}r$, if $\epsilon/2\le (9/10)(5/3)^k 6d^2\Lambda\dt$, then
the number of exponentials needed to simulate Hamiltonian evolution, using constant-sized time steps and within error $\epsilon/2$,
is bounded above by
\begin{equation}\label{eq:Nval}
N_{\exp}\le 2m5^{k-1}\left\lceil\frac{2\big[ 2k(5/3)^{k-1}6d^2\Lambdavar\dt\big]^{1+1/2k}}{(\epsilon/2)^{1/2k}}\right\rceil.
\end{equation}
Note that we have replaced $\epsilon$ with $\epsilon/2$, because we require error bounded by $\epsilon/2$ here.
To ensure that condition \eqref{eq:WBHS10condition} holds, we then replace $\epsilon$ by $\tilde\epsilon$,
which ensures that this condition holds and that the error is no greater than $\epsilon/2$.

Using $m=6Md^2$ and simplifying gives
\begin{equation}
	N_{\rm{exp}} \le 12Md^25^{k-1}
		\left\lceil 24kd^2\Lambdavar\dt\left(\frac{5}3\right)^k \left(\frac{6d^2\Lambdavar\dt}{\tilde\epsilon/2}\right)^{1/2k} \right\rceil.\label{eq:Nexpmod}
\end{equation}
The number of oracle queries that are needed in our simulation is
$\Noracle=\Costvar N_{\exp}$, which gives Eq.\ \eqref{eq:Noracle0}.

Finally, we verify the claim that the number of basis transformations is bounded above by $N_{\exp}/(3d^2)$ by counting the number of basis transformations that result from
using the Lie-Trotter-Suzuki formula.  As the BACS decomposition method expresses a \sparse{d} Hamiltonian as a sum of~$6d^2$
one-sparse Hamiltonians, and~$\{H_\alphavar\}$ is \sparse{d}, it follows that the Hamiltonian can be expressed as
\begin{equation}
H(t)=\sum_{\alphavar=1}^M T_\alphavar^{\dagger}\left( \sum_{j=1}^{6d^2} H_{\alphavar,j}(t) \right) T_\alphavar,\label{eq:BACSdecomp}
\end{equation}
where~$H_{\alphavar}(t)=\sum_{j=1}^{6d^2}H_{\alphavar,j}(t)$, and~$\{H_{\alphavar,j}\}$ is one-sparse. We can then use~\eqref{eq:suzuki2} to show that~$U_1(\initt+\tau,\initt)$ becomes
\begin{equation}
\left[\prod_{\alphavar=1}^M T_\alphavar^{\dagger}\left(\prod_{j=1}^{6d^2} \exp(-iH_{\alphavar,j}(\initt+\tau/2) \tau/2) \right) T_\alphavar \right] \left[ \prod_{\alphavar=M}^1 T_\alphavar^{\dagger}\left(\prod_{j=6d^2}^1 \exp(-iH_{\alphavar,j}(\initt+\tau/2) \tau/2) \right)T_\alphavar \right].\label{eq:u1examp}
\end{equation}
Equation \eqref{eq:u1examp} has only $4M$ basis transformations, but $12Md^2$ one-sparse operator exponentials.  Because $U_k$ is a product of~$5^{k-1}$ such approximations, there are at most $3d^2$ basis transformations per one-sparse operator exponential in $U_k$. Because the approximation to~$U$ in our decomposition is a product of~$U_k$~\cite{WBHS10}, \begin{equation}
N_{\rm{T}}\le N_{\exp}/(3d^2),\label{eq:thmbt0}
\end{equation} which implies that $N_{\rm{T}}\le \Noracle/(3\Costvar d^2)$ via this method.
\end{proof}

This Lemma shows how the complexity scales when the Hamiltonian is permitted to be a sum of terms that are each individually sparse in different bases.
This result only considers the error in integrator, and does not consider the contribution of the round-off error that occurs due to discretizing both time and the matrix elements of~$H_\alphavar$. The following Theorem
gives values of the precision that are sufficient to ensure this round-off error is $\epsilon/2$, implying
that the total error is at most $\epsilon$.

\begin{theorem}\label{thm:timedepsim}
If $\{H_\alphavar:\mathbb{R}\mapsto\mathbb{C}^{2^n\times 2^n};\alphavar=1,\ldots,M\}$ is a set of time-dependent Hermitian operators that is \sparse{d} and~$\Lambda$-$2k$-smooth
on~$[\initt,\initt+\dt]$, and~$\numBitsInT$ and~$\numQubitsInH$ satisfy~\eqref{eq:n'eq0}, then for any $\epsilon\in (0,1]$, the evolution generated by $H(t)=\sum_\alphavar T_\alphavar^{\dagger}H_\alphavar(t) T_\alphavar$ can be simulated within
error $\epsilon$, while using a number of oracle queries to~$\BBnoarg$ and~$\fynoarg$, $\Noracle$,  that are bounded above by
\begin{align}
\Noracle&\leq 12 \Costvar Md^2 5^{k-1} \left\lceil 24kd^2\Lambdavar\dt \left(\frac{5}{3} \right)^{k}\left(\frac{6d^2\Lambdavar\dt}{\tilde\epsilon/2}\right)^{1/2k}\right\rceil,
\end{align}
where~$\tilde{\epsilon}:= \min\{\epsilon,18(5/3)^{k-1}d^2\Lambdavar\dt\}$ and the number of queries to~$\{T_\alphavar\}$, $N_{\rm T}$, obeys $N_{\rm T}\le \Noracle/(3Cd^2)$.
\end{theorem}

\begin{proof}
The error in our simulation scheme arises from two sources: the discretization error, and the error due to the integrator.  By requiring that~$\numBitsInT$ and~$\numQubitsInH$ satisfy~\eqref{eq:n'eq0}, Lemma~\ref{lem:adError} implies that the error due to the discretization is no more than $\epsilon/2$. Note that, because~$\mathcal{I}$ is just a single time interval, the $\dt$ here corresponds to the $\dt$ in Lemma~\ref{lem:adError}, and the condition \eqref{eq:notshort} is automatically satisfied.
Using Lemma \ref{lem:timedepsim}, we find that the simulation can be performed such that the error in the integrator is no greater than $\tilde\epsilon/2$, and the query
complexity of the simulation satisfies the inequalities in Eqs.\ \eqref{eq:Noracle0} and \eqref{eq:thmbt0}. Because $\tilde\epsilon\le \epsilon$, using the triangle inequality shows that the total error is no greater than $\epsilon$.
\end{proof}

In some cases where~$\{H_\alphavar\}$ is not smooth, the Hamiltonian evolution can be simulated more efficiently by deleting a neighborhood from $[\initt,\initt+\dt]$ around each point where the derivatives of~$\{H_\alphavar\}$ diverge, and using the integrator $U_k$ to approximate the time-evolution in the remainder of the interval.
The query complexity for simulating Hamiltonian evolution by this method is given by the following corollary.
\begin{corollary}
\label{cor:piecewise}
Let $\{H_\alphavar:\mathbb{R}\mapsto\mathbb{C}^{2^n\times 2^n};\alphavar=1,\ldots,M\}$ be a set of time-dependent Hermitian operators that is $\Lambdavar$-$2k$-smooth on~$\mathcal{I}=(\initt,\initt+\dt)\setminus\{t_1,\ldots, t_L\}$, where~$\initt<t_1<\cdots<t_L<\initt+\dt$, with the additional conditions
\begin{enumerate}
\item $\exists~H_{\rm{max}}\in \mathbb{R}: H_{\max}\ge\max_{t\in[\initt,\initt+\dt]}\|H(t)\|$,
\item $0<\epsilon\le\min\{1, 27(5/3)^{k-1}d^2\Lambdavar\dt\}$,
\item $\numBitsInT$ and~$\numQubitsInH$ satisfy~\eqref{eq:n'eq0}, and
\item $\dt/2^{\numBitsInT}< \min_{\ell=0,\ldots,L}(t_{\ell+1}-t_\ell)$ with~$t_{L+1}:=\initt+\dt$.
\end{enumerate}
Then the query complexity for simulating evolution generated by $H(t)=\sum_\alphavar T_\alphavar^{\dagger}H_\alphavar(t) T_\alphavar$ within an error of~$\epsilon$ is
\begin{align}
\label{eq:cor3}
\Noracle&\le 12 \Costvar Md^2 5^{k-1}\left[(L+1)+ 24kd^2\Lambdavar \dt\left(\frac{5}{3} \right)^{k}\left(\frac{6d^2\Lambdavar\dt}{(\epsilon/3)}\right)^{1/2k}\right],
\end{align}
where $\Costvar$ is the number of oracle calls needed to simulate a one-sparse Hamiltonian, and $N_{\rm{T}}\le \Noracle/(3\Costvar d^2)$.
\end{corollary}
\begin{proof}
We remove $\delta$-neighborhoods around each~$t_\ell$ and simulate evolution
on the remaining time interval.
We then exploit the fact that~$H(t)$ is bounded to estimate the error incurred by omitting the evolution around those points.

We choose a value of~$\delta$ satisfying
\begin{equation}
\label{eq:defdel}
0<\delta\le\min\left\{ \frac{1}{2}\left[\min_{\ell=0,\ldots,L}(t_{\ell+1}-t_{\ell})-\dt/2^{\numBitsInT}\right], \frac {1}{2H_{\rm{max}}}\log\left(1+\frac{\epsilon/6}{L+2}\right)  \right\},
\end{equation}
where~$t_{L+1}:=\initt+\dt$.
The evolution operator is
\begin{align}
	 U(\initt+\dt,\initt)&=U(t_{L+1},t_{L+1}-\delta)\Biggr[\left(\prod_{\ell=L}^1 U(t_{\ell+1}-\delta,t_{\ell}+\delta)U(t_{\ell}+\delta,t_{\ell}-\delta)\right)
\times U(t_1-\delta,t_0+\delta)
	U(t_0+\delta,t_0)\label{eq:Udecomp2}\Biggr]\\
	&\approx \prod_{\ell=L}^0 U(t_{\ell+1}-\delta,t_{\ell}+\delta).\label{eq:Udecomp3}
\end{align}
For~$\epsilon\le 27(5/3)^{k-1}d^2\Lambdavar\dt$,
\begin{equation}
	2\epsilon(t_{\ell+1}-t_\ell-2\delta)/(3\dt)\le 18(5/3)^{k-1}d^2\Lambdavar(t_{\ell+1}-t_\ell-2\delta).
\end{equation}
With this restriction, using Lemma~\ref{lem:timedepsim} each of the $L+1$ evolutions in~\eqref{eq:Udecomp3} can be simulated with integrator error bounded above by $\epsilon(t_{\ell+1}-t_\ell-2\delta)/(3\dt)$ using no more than
\begin{equation}
	12 \Costvar Md^2 5^{k-1}\left\lceil 24kd^2\Lambdavar(t_{\ell+1}-t_{\ell}-2\delta)
	\left(\frac{5}{3} \right)^{k}\left(\frac{6d^2\Lambdavar\dt}{(\epsilon/3)}\right)^{1/2k}\right\rceil\label{eq:piecewise1}
\end{equation}
oracle queries.
Using Eq.~(4.69) of Ref.\ \cite{NC00}, which states that for unitary operators,
\begin{equation}
\label{eq:mike}
	\left\|\prod_j U_j - \prod_k V_k\right\|\le \sum_j \|U_j-V_j\|,
\end{equation}
the total error is bounded above by
\begin{equation}
\sum_{\ell} \epsilon(t_{\ell+1}-t_\ell-2\delta)/(3\dt)<\epsilon/3.
\end{equation}
Then, summing~\eqref{eq:piecewise1} over $\ell$ gives inequality \eqref{eq:cor3} as an upper bound for the number
of oracle queries made. The additional factor of~$(L+1)$ in Eq.~\eqref{eq:cor3} is to take account of the ceiling function in Eq.~\eqref{eq:piecewise1}.
The bound for the number of basis transformations is obtained by using~$N_{\rm T}=N_{\exp}/3d^2$, and by summing over all $L+1$ subintervals.

To bound the overall error, we need to bound the error due to approximating~\eqref{eq:Udecomp2} with~\eqref{eq:Udecomp3}.
Because $H(t)$ is bounded on~$[\initt,\initt+\dt]$, and~$H_{\rm{max}} \ge \max_{t\in[\initt,\initt+\dt]}\|H(t)\|$, the unitary evolution over each~$t_\ell$ for $\ell=1,\ldots,L$ satisfies
\begin{equation}
	\left\| U(t_\ell+\delta,t_{\ell}-\delta)-\openone \right\|\le e^{2H_{\rm{max}}\delta}-1.\label{eq:deltabd1}
\end{equation}
For $\ell=0$ and~$L+1$ we have
\begin{align}
	\left\| U(t_0+\delta,t_0)-\openone \right\| &\le e^{H_{\rm{max}}\delta}-1, \nn
	\left\| U(t_{L+1},t_{L+1}-\delta)-\openone \right\| &\le e^{H_{\rm{max}}\delta}-1.
\end{align}

These errors can be made suitably small by using the restriction on~$\delta$ specified in Eq.~\eqref{eq:defdel}.
The first expression in the braces in \eqref{eq:defdel} ensures that the inequality $\delta<\min_{\ell=0,\ldots,L}(t_{\ell+1}-t_{\ell})/2$ is satisfied.
The second ensures that the error in approximating the evolution about each of the $t_\ell$ by $\openone$ is bounded above by $\epsilon/[6(L+2)]$.
Using Eq.~\eqref{eq:mike}, the total error in approximating~\eqref{eq:Udecomp2} with~\eqref{eq:Udecomp3} is bounded above by $\epsilon/6$.
Combining this with the bound on the error of~$\epsilon/3$ for the integrators over the time intervals $[t_{\ell}+\delta,t_{\ell+1}-\delta]$, the total error due to omitting the $\delta$-neighborhoods from the
simulation and using the Lie-Trotter-Suzuki formula is bounded above by $\epsilon/2$.

We now use Lemma~\ref{lem:adError} to ensure that the roundoff error is also bounded above by $\epsilon/2$.
The definition of~$\dt$ used in that Lemma gives the $\dt$ used in this corollary, and so may be used in the restrictions without change.
The restriction~$\dt/2^{\numBitsInT}< \min_{\ell}(t_{\ell+1}-t_\ell)$ and the choice of~$\delta$ in~\eqref{eq:defdel} ensure that the restriction \eqref{eq:notshort} of Lemma~\ref{lem:adError} holds.
We have also required that~$\numBitsInT$ and~$\numQubitsInH$ satisfy~\eqref{eq:n'eq0}, so all conditions required for Lemma~\ref{lem:adError} hold, and the round-off error may be bounded by $\epsilon/2$. As the error in the integrator has also been bounded by $\epsilon/2$, the total error is no greater than $\epsilon$.
\end{proof}

We have shown in this section that time-dependent Hamiltonian evolutions can be simulated by using the product formula approach to simulation, even if there are discontinuities in the Hamiltonian.
If the Hamiltonian is sufficiently smooth, then these simulations achieve the same near-linear scaling as the BACS
simulation achieves.
In the next section we improve upon these results by presenting a method that uses adaptive time steps.

\section{Adaptive Decomposition Scheme}
\label{sec:adaptive}
We saw in the previous section that if we use time steps that have constant size in our simulation algorithm, then the number of oracle calls used to simulate Hamiltonian evolution depends on the largest values of the norms of $H(t)$ and its derivatives.
For Hamiltonians whose time-dependence is sharply peaked, the maximum values of these quantities can be quite large relative to their time-averages.
It is natural to ask if the complexity can be made to depend on the average values, rather than the maximum values, by using adaptive time steps.
In this section we show that this is indeed possible.

As discussed in Sec.\ \ref{sec:metad}, this result is nontrivial because of the interdependence of the different parameters.
We choose a sequence of times $\{t_p\}$ such that~$\initt<t_1<\cdots<t_r=\initt+\dt$.
These times are selected such that the error from using~$U_k$ is bounded above by $\epsilon/(2r)$ for each interval.
Given a function $\Upsilonfcn$, such that~$\{H_\alphavar\}$ is $\Upsilonfcn$-$2k$-pointwise-smooth on~$[\initt,\initt+\dt]$, the size of the interval will be inversely proportional to the maximum value of $\Upsilonfcn(t)$ in that interval.
The exact result is given in the following Lemma.
\begin{lemma}
\label{lem:inter}
Let $\{H_\alphavar:\alphavar=1,\ldots M\}$, where~$H_{\alphavar}:\mathbb{R}\mapsto\mathbb{C}^{2^n\times 2^n}$, be a set of time-dependent Hermitian operators that is $\Upsilonfcn$-$2k$-pointwise-smooth and \sparse{d} on~$[\initt,\initt+\dt]$, and~$\{t_1,\ldots,t_r\}$ be a set of~$r$ moments in time such that~$\initt<t_1<\cdots<t_r=\initt+\dt$. If
\begin{equation}
\label{eq:varstep}
	 \left(\max_{t\in[t_p,t_{p+1}]}\Upsilonfcn(t)\right)(t_{p+1}-t_p)\leq \frac{(\epsilon/r)^{1/(2k+1)}}{24d^2k(5/3)^{k-1}},
\end{equation}
where~$\epsilon\in(0,1]$, then
\begin{equation}
\label{eq:udif3}
\left\|U(\initt+\dt,\initt)-\prod_{p=1}^rU_k(t_{p},t_{p-1})\right\|\le \epsilon/2.
\end{equation}
\end{lemma}
\begin{proof}
Using Theorem 3 and Eq.~(A.3) of Ref.\ \onlinecite{WBHS10}, the error in each interval is bound above by
\begin{equation}
\|U(t_p+\dt_p,t_p)-U_k(t_p+\dt_p,t_p) \|\le2[2k(5/3)^{k-1}\Lambda\Delta t_p]^{2k+1},\label{eq:suzukierrbd}
\end{equation}
where~$\dt_p:=t_{p+1}-t_p$, provided that
\begin{equation}4k\sqrt{2}(5/3)^{k-1}\Lambda\dt_p\le 3/2,\label{eq:thm3condition}\end{equation}
for a set of operators that is $\Lambda$-$2k$-smooth.
As $\Lambda$ is bounded above by $6d^2\max_{t\in[t_p,t_{p+1}]}\Upsilonfcn(t)$, the upper bound in~\eqref{eq:suzukierrbd} becomes
\begin{equation}
\label{eq:udif}
\|U(t_p+\dt_p,t_p)-U_k(t_p+\dt_p,t_p) \|\le2\left[12d^2k(5/3)^{k-1}\max_{t\in[t_p,t_{p+1}]}\Upsilonfcn(t)\Delta t_p\right]^{2k+1}.
\end{equation}
Because $\epsilon\le 1$ and~$r\ge 1$, the condition \eqref{eq:varstep} implies that~\eqref{eq:thm3condition} is satisfied, and therefore that Eq.~\eqref{eq:udif} holds.

Using Eq.~\eqref{eq:varstep} in Eq.~\eqref{eq:udif} gives
\begin{equation}
\label{eq:udif2}
\|U(t_p+\dt_p,t_p)-U_k(t_p+\dt_p,t_p) \|\le \epsilon/(2r).
\end{equation}
Then, using Eq.~\eqref{eq:mike} the error due to using~$r$ Lie-Trotter-Suzuki integrators is at most $\epsilon/2$, hence yielding~\eqref{eq:udif3}.
\end{proof}

There are two challenges in using this result as a method for choosing the time steps.
First, Eq.\ \eqref{eq:varstep} depends on the maximum value of $\Upsilonfcn(t)$ over the interval, which may be difficult to find, particularly because it depends on the size of the time interval.
Ideally, it should depend only on the value at time $t_p$, in order to give a simple method of determining the time interval.
Second, Eq.\ \eqref{eq:varstep} depends on the number of time steps $r$, which is not known in advance.
It is only known once all time intervals have been determined, but we wish to determine these via Eq.\ \eqref{eq:varstep}.

To break these interdependencies, we first need a bound on the derivative of $\Upsilonfcn(t)$, so we can bound the value of $\Upsilonfcn(t)$ on the interval by the value at time $t_p$.
It is not obvious what bound should be taken, but it can be shown that if $\Upsilonfcn_P$ is the smallest possible function such that~$\{H_\alphavar\}$ is $\Upsilonfcn$-$P$-pointwise-smooth, then it will satisfy $|\partial_t\Upsilonfcn_P(t)| \le [\Upsilonfcn_{P+1}(t)]^2$ (see Appendix \ref{appendix:derivbd}).
In many cases we can expect that there exists a constant $\scalconst$ such that $\Upsilonfcn_{P+1}(t)\le \scalconst\Upsilonfcn_{P}(t)$.
This motivates taking the condition $|\partial_t\Upsilonfcn(t)| \le \scalconst^2[\Upsilonfcn(t)]^2$.

Second, we provide a quantity $r_g$, given in Eq.\ \eqref{eq:rdef}, that is an upper bound on the number of time steps.
If the time steps are chosen according to Eq.\ \eqref{eq:varstep} with $r_g$ instead of $r$, then Eq.\ \eqref{eq:varstep} must still hold.
Using this approach, we obtain the following Theorem.

\begin{theorem}\label{thm:adaptiveResult2}
If $\{H_\alphavar:\mathbb{R}\mapsto\mathbb{C}^{2^n\times 2^n};\alphavar=1,\ldots,M\}$ is a set of time-dependent Hermitian operators that is $\Upsilonfcn$-$2k$-pointwise-smooth, $d$-sparse on~$\mathcal{I}=[\initt,\initt+\dt]$, $\numBitsInT$,  and~$\numQubitsInH$ satisfy~\eqref{eq:n'eq0}, and there exists $\scalconst\in\mathbb{R}$ such that
$|\partial_t\Upsilonfcn(t)| \le \scalconst^2[\Upsilonfcn(t)]^2$ $\forall~t\in\mathcal{I}$, then for any $\epsilon\in (0,1]$, the evolution generated by $H(t)=\sum_\alphavar T_\alphavar^{\dagger}H_\alphavar(t) T_\alphavar$
can be simulated within error $\epsilon$, with %values of~$\Noracle$ and~$N_{\rm T}$ that satisfy
\begin{equation}
\Noracle \le 12\Costvar Md^25^{k-1}\left\lceil\frac{\left[24d^2k(5/3)^{k-1}\overline\interval(\initt+\dt,\initt)\dt \right]^{1+1/2k}} {(\epsilon/4)^{1/2k}} +3\scalconst^2\overline\interval (\initt+\dt,\initt)\dt+1\right\rceil\label{eq:Noracleresult2},
\end{equation}
where~$\log^*$ is the iterated logarithm function,~$\upp$ is given in Lemma~\ref{lem:adError}, $\Costvar$ is the number of oracle calls needed to simulate a one-sparse Hamiltonian,
and $N_{\rm T}\le \Noracle/(3\Costvar d^2)$.
\end{theorem}
\begin{proof}
For this proof, we choose a set of~$r$ times $\{t_1,\ldots,t_r\}$, such that~$\initt<t_1<\cdots<t_r=\initt+\dt$.
We define
\begin{equation}
\Avar := \frac{[24d^2k(5/3)^{k-1}\overline\interval (\initt+\dt,\initt)\dt]^{1+1/2k}}{\epsilon^{1/2k}},\label{eq:Adef}
\end{equation}
and
\begin{equation}
 r_g := \lceil \Avar+3\scalconst^2\overline\interval (\initt+\dt,\initt)\dt+1 \rceil. \label{eq:rdef}
\end{equation}
We choose the times via the recurrence relation, for $p=0,\ldots,r-1$,
\begin{equation}
\label{eq:recur}
t_{p+1} = t_p + \frac 1{\Upsilonfcn(t_p)} \frac 1{\upb+\scalconst^2},
\end{equation}
where
\begin{equation}
\upb := \frac{24d^2k(5/3)^{k-1}}{(\epsilon/r_g)^{1/(2k+1)}}\label{eq:varstep2}.
\end{equation}
We will show shortly that~\eqref{eq:recur} implies that~\eqref{eq:varstep} is held, thereby guaranteeing
that the integrator error is bounded above by $\epsilon/2$.  In order to prove this we first need to
find an upper bound for $\Upsilonfcn(t+\delta t)$, for small values of~$\delta t>0$, in terms of~$\Upsilonfcn(t)$.

 We find this bound by solving the differential equations given by $\Upsilonfcn$ saturating the inequality $|\partial_t\Upsilonfcn(t)|\le \scalconst^2[\Upsilonfcn(t)]^2$.
In particular, for $\delta t\ge 0$, we have
\begin{equation}
\label{eq:lbnds}
\frac {\Upsilonfcn(t)}{1+\scalconst^2\Upsilonfcn(t)\delta t}\le
\Upsilonfcn(t+\delta t) \le \frac {\Upsilonfcn(t)}{1-\scalconst^2\Upsilonfcn(t)\delta t}.
\end{equation}
By using the inequality on the left, and the fact that the minimum of a function is less than or equal to its average value, we find
\begin{align}
	\overline\interval(t_{p+1},t_p) &\ge \frac{\Upsilonfcn(t_p)}{1+\scalconst^2\Upsilonfcn(t_p)\dt_p}.
\label{eq:lbbnd}
\end{align}
This, together with recurrence relation \eqref{eq:recur}, yields
\begin{equation}
	1 \le \overline\interval(t_{p+1},t_p)\dt_p \left( \upb + 2\scalconst^2\right).\label{eq:adaptive1}
\end{equation}
Summing over the first $r-1$ subintervals then gives
\begin{align}
r-1 &\le \sum_{p=0}^{r-2}\overline\interval(t_{p+1},t_p)\dt_p \left( \upb + 2\scalconst^2\right)\nonumber\\
&=(t_{r-1}-\initt)\overline\interval(t_{r-1},\initt)(Y+2K^2).
\end{align}
By solving the above equation for $r$, and using the fact that~$(t_{r-1}-\initt)\overline\interval(t_{r-1},\initt)\le\dt\overline\interval(\initt+\dt,\initt)$, we
have that
\begin{align}
r &\le\dt\overline\interval(\initt+\dt,\initt) \left( \upb + 2\scalconst^2\right)+1.\label{eq:adaptive2}
\end{align}
As $r$ is an integer, it is also bounded above by
\begin{equation}
r \le\left\lfloor\dt\overline\interval(\initt+\dt,\initt) \left( \upb + 2\scalconst^2\right)+1\right\rfloor.\label{eq:adaptive2p5}
\end{equation}
This also implies that for any $\gamma>0$,
\begin{equation}
r\le\left\lceil \dt\overline\interval(\initt+\dt,\initt) \left( \upb + 2\scalconst^2\right) +\gamma\right\rceil.\label{eq:adaptive3}
\end{equation}
If we take~$\gamma=1/3$ in~\eqref{eq:adaptive3} and use the definitions of~$A$ and~$Y$, we find
\begin{align}
r &\le \left\lceil A^{2k/(2k+1)}r_g^{1/(2k+1)} + 2\scalconst^2\overline\interval(\initt+\dt,\initt)\dt+1/3\right\rceil \nn
&= \left\lceil A[1+(r_g/A-1)]^{1/(2k+1)} + 2\scalconst^2\overline\interval(\initt+\dt,\initt)\dt+1/3\right\rceil \nn
&\le \left\lceil A+\frac{1}{2k+1}(r_g-A) + 2\scalconst^2\overline\interval(\initt+\dt,\initt)\dt+1/3\right\rceil \nn
&\le \left\lceil A+\frac{1}{3}(3\scalconst^2\overline\interval(\initt+\dt,\initt)\dt+2) + 2\scalconst^2\overline\interval(\initt+\dt,\initt)\dt+1/3\right\rceil \nn
&\le \left\lceil A+ 3\scalconst^2\overline\interval(\initt+\dt,\initt)\dt+1\right\rceil = r_g.
\end{align}

The inequality on the right of Eq.~\eqref{eq:lbnds} gives
\begin{equation}
\max_{t\in[t_p,t_{p+1}]}\Upsilonfcn(t) \le \frac {\Upsilonfcn(t_p)}{1-\scalconst^2\Upsilonfcn(t_p)\dt_p}.
\end{equation}
This, together with the recurrence relation \eqref{eq:recur}, gives
\begin{equation}
\max_{t\in[t_p,t_{p+1}]}\Upsilonfcn(t)\dt_p \le 1/\upb.
\end{equation}
Because $r\le r_g$, the condition \eqref{eq:varstep} of Lemma \ref{lem:inter} is satisfied, and the error from the integrator is no more than $\epsilon/2$.
This implies that the total error is bounded above by $\epsilon$ if $\numBitsInT$ and~$\numQubitsInH$ are chosen as in Lemma~\ref{lem:adError}.

We can place an upper bound on the number of exponentials used in the approximation by multiplying $r$ by the number of exponentials in each \Th{k}-order Lie-Trotter-Suzuki approximation.
We use $r_g$ as an upper bound for $r$, note that there are $12Md^25^{k-1}$ one-sparse operator exponentials in each of the \Th{k}-order Lie-Trotter Suzuki approximations, and use $\Noracle=\Costvar N_{\exp}$, giving Eq.\ \eqref{eq:Noracleresult2}.
The bound on $N_{\rm T}$ follows from the fact that $\Noracle=\Costvar N_{\exp}$, and $N_{\rm T}\le N_{\exp}/3d^2$.
\end{proof}

The result of Theorem~\ref{thm:adaptiveResult2} shows that adaptive time steps can be used to dramatically reduce the complexity of simulating certain time-dependent Hamiltonian evolutions.  This improvement stems from the fact that the performance of the constant time step algorithm depends on the largest possible value of $\Upsilonfcn$, whereas the performance of the adaptive version scales with the average value of $\Upsilonfcn$.  In some cases, such as an example in the following section, the average can be much smaller than the largest value of $\Upsilonfcn$, leading to a substantial difference in performance.

\section{Examples}
We now examine the performance of our adaptive time step decomposition method for a pair of examples.
In the first example, we examine a Hamiltonian that is a Gaussian approximation to a Dirac-delta function.
In our second example, we examine a non-analytic Hamiltonian, and show that we can substantially reduce the number of exponentials used in some simulations by choosing~$k$ adaptively as well.

\begin{figure}[t!]
\begin{minipage}[t]{0.4\linewidth}

\includegraphics[scale=0.7]{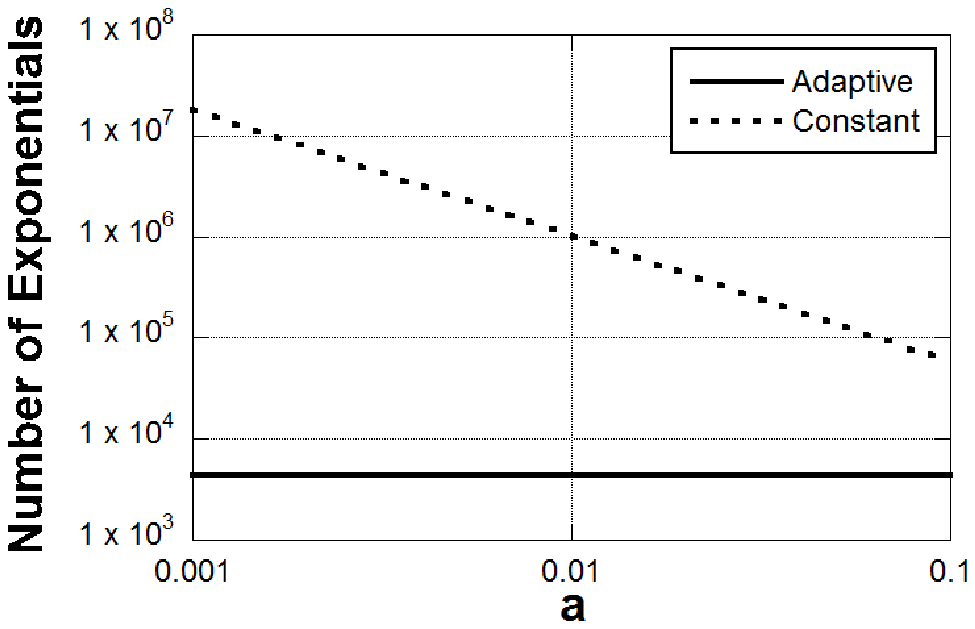}
\caption{This figure shows a comparison between the number of exponentials used in
 our adaptive and constant step size methods for $H(t)=\exp(-(t-1)^2/a^2)/(a\sqrt{\pi})\openone$, $t\in[0,2]$,
$\epsilon=10^{-4}$ and~$k=2$.}
\label{fig:adaptivevsnon}
\end{minipage}%
\hspace{1cm}
\begin{minipage}[t]{0.4\linewidth}

\includegraphics[scale=0.7]{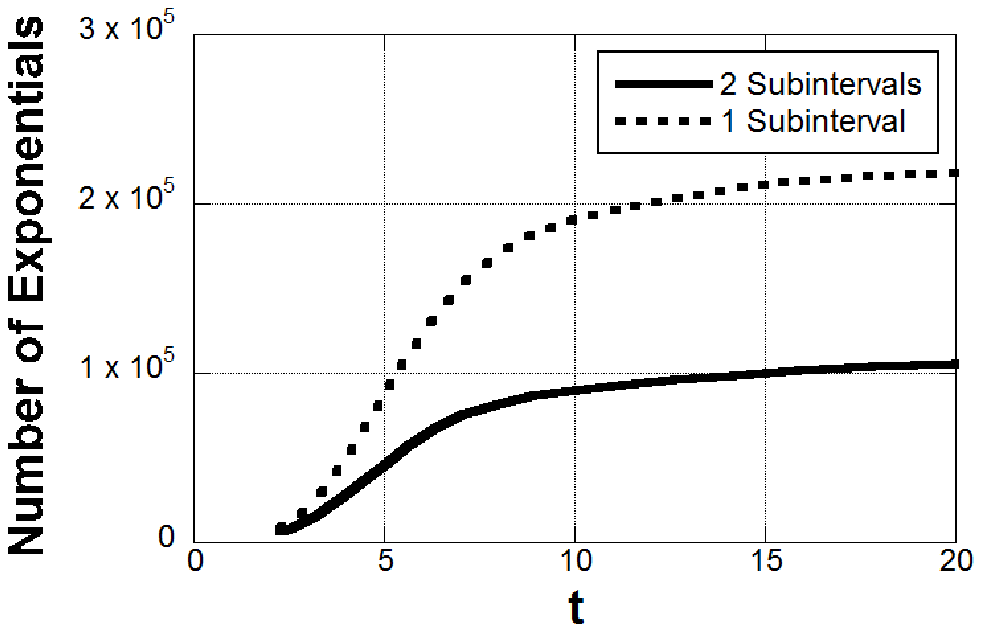}
\caption{This plot shows the number of exponentials found by choosing~$k$ adaptively for $H(t)=t^5\sin(1/t)\exp(-t)\openone$, and~$\epsilon=10^{-4}$ for one and two subintervals with~$k=1$ and~$k=2$ respectively.}
\label{fig:2layer}
\end{minipage}
\end{figure}
The Hamiltonian that we choose in our first example is, for $a>0$,
\begin{equation}
H(t)=\exp(-(t-1)^2/a^2)/(a\sqrt{\pi})\openone.
\end{equation}
We compare the cost by plotting $N_{\exp}$ as a function of~$a$ for both methods.
We choose our approximations to be $2^{\textup{nd}}$-order Lie-Trotter-Suzuki formulae~$(k=2)$ with~$\epsilon=10^{-3}$.
For our adaptive method, rather than choosing the time steps using the upper bound on $r_g$ as in Theorem \ref{thm:adaptiveResult2}, we use an iterative process to find the number of steps.
That is, we choose each step using \eqref{eq:varstep} with an initial guess of $r=r_g$, then find a sequence of steps via \eqref{eq:varstep}, and then count the number of steps to find a new guess for $r$.
Iterating this process gives the solution for $r$.

Fig.~\ref{fig:adaptivevsnon} shows that the adaptive
method results in a value of~$N_{\exp}$ that is approximately
constant in $a$.  In comparison, the number of exponentials used in the constant step size case diverges as~$a$ approaches
zero.  This divergence occurs because~$\lim_{a\rightarrow 0} \max_{t\in[0,2]}\Upsilonfcn_{4}(t)=\infty$, whereas the adaptive method
yields a nearly constant value of~$N_{\exp}$, because
$\int_{0}^{2}\Upsilonfcn_4(t)\mathrm{d}t$ only weakly depends on~$a$.
This shows that it can be advantageous to choose the step size adaptively if $H(u)$ varies substantially over the interval of simulation.

The Hamiltonian that we use in our second example is
\begin{equation}
H(t)=t^5\sin(1/t)\exp(-t)\openone.
\end{equation}
The second-order Lie-Trotter-Suzuki formula should not be used as an approximation to~$U(\dt,0)$~\cite{WBHS10},
because the third derivative of~$H(t)$ diverges near $t=0$.  However, $U_2$ can be used to approximate
the time-evolution on any closed interval that excludes $t=0$.  A natural way to handle this is to divide the
approximation into two subintervals, and use different values of~$k$ to approximate the evolution within each subinterval.  The evolution in the first subinterval, $u\in[0,t']$, is approximated using~$U_1$, whereas the
evolution in the remaining subinterval is approximated using~$U_2$.  We reduce the number of
exponentials that are used in our simulation schemes for any fixed value of~$t'$ by applying our adaptive step size algorithm in each subinterval.
We then vary $t'$ using a gradient search method to change the size of each subinterval to further reduce the number of exponentials that are used in the simulation.

We show in Fig.~\ref{fig:2layer} that choosing~$k$ adaptively can lead to reductions in the number
of exponentials that are needed to approximate the time-evolution operator. Therefore, when simulating the evolution generated by  Hamiltonians with
singularities, it may be more efficient to use lower-order formul\ae~near the singularity, and higher-order
formul\ae~further away from it.  This method may have applications in situations where high-order integrators
fail to provide the scaling expected for singular differential equations, such as those that occur in Coulomb problems~\cite{Chi06}.

\section{Simulating One-Sparse Hamiltonians}\label{sec:oracles}
In Sections~\ref{sec:timedep} and~\ref{sec:adaptive} we specified the cost of simulating the evolution as a function of $\Costvar$, which is the number of oracle calls that are needed to simulate a one-sparse Hamiltonian.  In this section, we provide an upper bound for $\Costvar$ and discuss how this cost relates to those obtained using other oracle definitions.  We then use this bound for $\Costvar$ in concert with our results from Sec.~\ref{sec:adaptive} to prove Theorem~\ref{thm:adaptiveResult} and Corollary~\ref{cor:linscale}.

\begin{lemma}
\label{lem:oraclem}
Let $\{H_\alphavar:\alphavar=1,\ldots,M\}$, where~$H_\alphavar:\mathbb{R}\mapsto\mathbb{C}^{2^n\times 2^n}$,  be a set of time-dependent Hermitian operators that is $2k$-smooth on~$[\initt,\initt+\dt]$, and let $\{H_{\alphavar,j}\}$ represent the set of one-sparse Hamiltonians that
result from applying the BACS decomposition algorithm to each~$\{H_\alphavar\}$.
The query complexity, $\Costvar$, of simulating $\exp\left\{-iH_{\alphavar,j}(\tau)\delta t\right\}$ for any $j$ and~$[\tau,\tau+\delta t]\subseteq [\initt,\initt+\dt]$ using the oracles
$\fynoarg$ and~$\BBnoarg$ and~$\numQubitsInH$
  bits of precision to represent the matrix elements, is bounded above by
\begin{equation}
\label{eq:cbound}
	\Costvar \le 4n(z_n+2)+3\numQubitsInH,
\end{equation}
where~$z_n$ is the number of times the mapping $z\mapsto \lceil2\log_2(z) \rceil$ must be iterated, starting at $n$, before reaching a value that is at most $6$.
\end{lemma}

To prove this Lemma, we take advantage of the polar decomposition of the matrix elements of the Hamiltonians.
This enables a remarkably efficient simulation of the evolution under one-sparse operator exponentials using a sequence of rotations that
are controlled by only one qubit at a time.
In particular, the result is as in the following Lemma.

\begin{lemma}\label{lem:onequbit}
Given that the assumptions of Lemma~\ref{lem:oraclem} are met, and if $\numQubitsInH$ bits of precision are used to
represent each of the matrix elements of~$H_j$, then~$\exp(-iH_{\alphavar,j}(\tau)\delta t)$ can be simulated using
one qubit to store these matrix elements and~$3\numQubitsInH$ queries to the oracle~$\BBnoarg$.
\end{lemma}

\begin{proofof}{Lemma~\ref{lem:onequbit}}
The decomposition scheme of BACS \cite{BACS07} takes a given row number $x$, and determines if there will be a matrix element in row $x$ assigned to $H_{\alphavar,j}$.
If there is, then the boolean function $\xi(x)$ is set to 1; otherwise it is zero.
If there is a matrix element, then the column number is determined, and $M_x$ and $m_x$ are set equal to the row and column numbers, such that $M_x\ge m_x$.
The decomposition method also generates a three-bit string $\nu(x)$, which we will not otherwise use in the simulation scheme.
The BACS decomposition method will therefore transform the initial state $\ket{\psi}$ according to
\begin{equation}
\ket{\psi}=\sum_x a_x\ket{x,0^{\otimes 2n+4}}\mapsto\sum_x a_x\ket{x,m_x,M_x,\nu(x),\xi(x)}.
\end{equation}

Then given this transformed state, our next goal is to perform a mapping between the subspace $\text{span}\{\ket{m_x},\ket{M_x}\}$ and an ancilla qubit space.  The purpose of this is to map this two-dimensional subspace onto one that can be evolved
using single-qubit operations.  This transformation is
\begin{equation}
\sum_{x}a_x \ket{x,m_x,M_x,\nu(m_x),\xi(m_x),0,0}\mapsto\sum_{x}a_x\ket{0^{n}}\otimes\left\{\begin{array}{cl}x=M_x=m_x,&\ket{m_x\oplus M_x,M_x,\nu(m_x),\xi(x),1,1}\\
x=M_x\neq m_x,&\ket{m_x\oplus M_x,M_x,\nu(m_x),\xi(x),1,0}\\
x= m_x\neq M_x,&\ket{m_x\oplus M_x,M_x,\nu(m_x),\xi(x),0,0}
\end{array}\right..\label{eq:subspacemap}
\end{equation}
The second last qubit is the one that encodes the two-dimensional subspace; it is $\ket{1}$ if $x=M_x$, and $\ket{0}$ if $x=m_x$.
The last qubit is used to indicate if $x$ is a member of a one-dimensional subspace.

\begin{figure}[t!]
\centering
%\scalebox{0.7}
\includegraphics[width=\columnwidth]{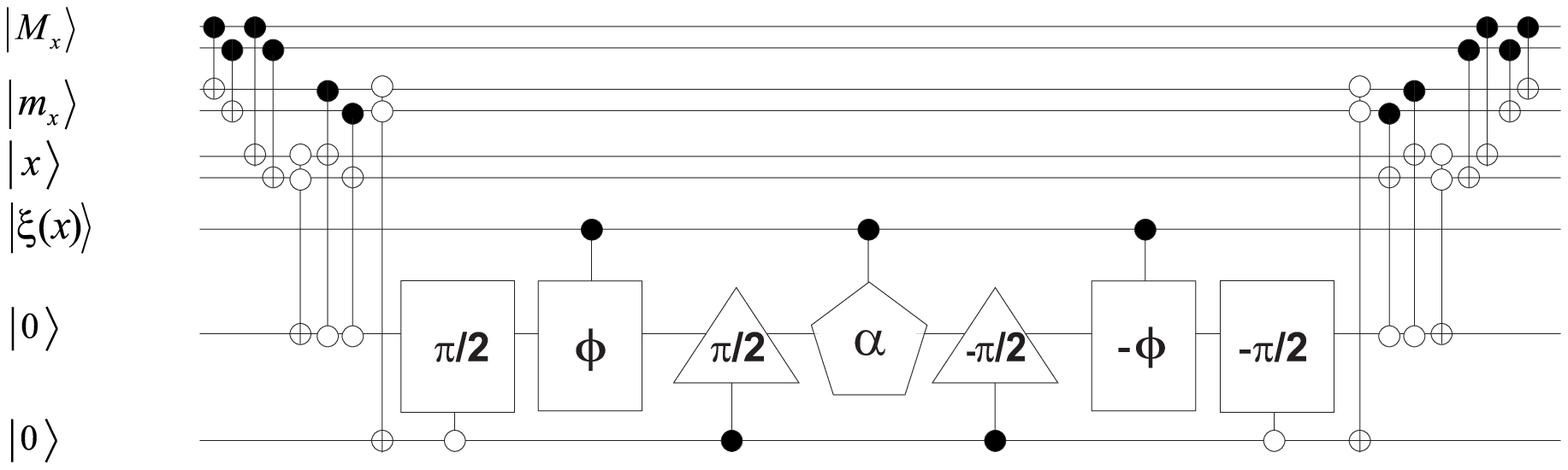}
\caption{\label{fig:simcircuit}This circuit simulates $\exp(-i\sparseham_{\alphavar,j}(t_p)\delta t_p)$ for the one-sparse Hamiltonian
$\sparseham_{\alphavar,j}$, given an input state
of the form of the LHS of~\eqref{eq:subspacemap}.  Here the variable $\phi=\text{Arg}\big([\sparseham_{\alphavar,j}(t_p)]_{m_x,M_x}\big)$ and~$\alpha=2|[\sparseham_{\alphavar,j}(t_p)]_{m_x,M_x}|\dt_p$.  Here we also use rectangles to represent $R_z$ rotations by a fixed angle,
triangles represent $R_x$ rotations and the pentagon represents a $R_y$ rotation.  These rotations can be enacted by querying the oracle
$\BBnoarg$ and performing controlled rotations on the output.}
\end{figure}
Given a state in the form of the RHS of~\eqref{eq:subspacemap}, we can simulate the evolution of~$\ket{x}$ by evolving the second last ancilla qubit,
whose state is logically equivalent to~$\ket{x}$.  To do so, we must know whether $x$ is in a one- or two-dimensional irreducible
subspace.
Specifically, the evolution operator takes one of two possible forms on the subspace $\text{span}(\ket{m_x},\ket{M_x})$. It performs the transformation
\begin{equation}\label{eq:1sparseonesubspace}
a_{M_x}\mapsto a_{M_x}\exp(-i[\sparseham_{\alphavar}(t_p)]_{m_xM_x}\dt)
\end{equation}
or
\begin{equation}\label{eq:1sparsetwosubspaces}
\left[\begin{array}{c}a_{m_x}\\a_{M_x}\end{array}\right]\mapsto \exp\left(-i\left[\begin{array}{cc} 0&[\sparseham_{\alphavar}(t_p)]_{m_xM_x}\\\big([\sparseham_{\alphavar}(t_p)]_{m_xM_x}\big)^{*}&0\end{array} \right]\dt\right)\left[\begin{array}{c}a_{m_x}\\a_{M_x}\end{array}\right],
\end{equation}
if the subspace is one- or two-dimensional, respectively.

We can write the two-dimensional rotation as a sequence of Pauli-$z$ and -$y$ rotations using standard decomposition techniques~\cite{NC00}.
Given $H_{m_x,M_x}=\rho\exp(i\phi)$, the resulting decomposition of the two-dimensional transformation is
\begin{equation}\label{eq:2d}
\exp\left(-i\left[\begin{array}{cc} 0&\big[\sparseham_{\alphavar}(t_p)\big]_{m_xM_x}\\\big[\sparseham_\alphavar(t_p)\big]_{m_xM_x}^{*}&0\end{array} \right]\dt_p\right)=R_z(-\pi/2)R_z(-\phi)R_y(2\rho\Delta t_p)R_z(\phi)R_z(\pi/2).
\end{equation}
The sequence of exponentials in \eqref{eq:2d} can be implemented using a sequence of Pauli rotations that are controlled by
the qubits that encode the matrix element $H_{m_x,M_x}$.

The one-dimensional subspace is evolved according to
\begin{equation}
\ket{x}\rightarrow \exp\biggr({-i\big[\sparseham_{\alphavar}(t_p)\big]_{xx}\dt_p}\biggr)\ket{x}.
\end{equation}
Because we have encoded $\ket{x}$ in the one-dimensional case to be $\ket{1}$ in \eqref{eq:subspacemap},
the one-dimensional transformation can be expressed as
$R_z\big(-2[\sparseham_{\alphavar}(t_p)]_{M_xM_x}\dt_p\big)$.
This rotation can also be written as,
\begin{equation}
R_z\big(-2[\sparseham_{\alphavar}(t_p)]_{M_xM_x}\dt_p\big)=R_z(-\phi)R_x(-\pi/2)R_y(2\rho\Delta t_p)R_x(\pi/2)R_z(\phi)\label{eq:1d},
\end{equation}
which allows us to write the one-dimensional rotation in a form that is similar to the two-dimensional rotation.

We present a circuit in Fig.\ \ref{fig:simcircuit} that enacts both the transformation in \eqref{eq:subspacemap} and also \eqref{eq:1d} and \eqref{eq:2d} coherently on each subspace.
We combine the rotations for both the one- and two-dimensional cases together in a single sequence of rotations, by making the $\pi/2$ rotations controlled by the last qubit.
In addition, we allow the rotations to be controlled by $\xi(x)$, so no rotations are performed if no matrix element has been assigned to row $x$ of $H_{\alphavar,j}$.

The rotations $R_y(\alpha)$, where~$\alpha=2\rho\dt_p$, and~$R_z(\phi)$ in Fig.\ \ref{fig:simcircuit} can each be implemented by calling the oracle
$\BBnoarg$ $\numQubitsInH/2$ times, provided $\numQubitsInH$ is even, and equal numbers of bits are used to encode the modulus and phase of the matrix element.  Since there are three rotations of this form, $3\numQubitsInH/2$ oracle
calls can be used to enact them.  However, in order to re-use the ancilla bits that record these values,
we need to call the oracle another $3\numQubitsInH/2$ times.  Therefore, the total number of calls made to~$\BBnoarg$
is bounded above by $3\numQubitsInH$.
\end{proofof}

Note that, in Fig.\ \ref{fig:simcircuit} we can access each qubit of the oracles independently, without needing to store the other qubits.
This is because, for each controlled operation (such as $R_z(\phi)$), we can call the oracle for one qubit of precision, perform the rotation for that qubit, then call the oracle again to erase the value, before calling the oracle for the next qubit.
Now that we have proven this Lemma, the proof of Lemma \ref{lem:oraclem} is simple.

\begin{proofof}{Lemma~\ref{lem:oraclem}}
The one-sparse matrix exponentials may be performed by using the BACS decomposition technique, then performing the rotations as described in the proof of Lemma~\ref{lem:onequbit}, then inverting the BACS decomposition technique to restore the ancilla qubits to their original states.
The BACS decomposition technique uses $2(z_n+2)$ queries to their oracle to identify the irreducible subspace that a basis state is in, and store this information in a qubit string~\cite{BACS07}.
Using an oracle that only provides one qubit at a time, the number of oracle calls is multiplied by a factor of $n$.
Another factor of 2 is obtained because the BACS decomposition technique is inverted, yielding a total number of queries of $4n(z_n+2)$.
Using Lemma~\ref{lem:onequbit}, the rotations can be performed using~$3\numQubitsInH$ calls to~$\BBnoarg$, so the total number of oracle calls needed is bounded as in Eq.\ \eqref{eq:cbound}.
\end{proofof}

Using Lemma \ref{lem:oraclem}, it is now straightforward to prove Theorem~\ref{thm:adaptiveResult}, which is our main result in the paper.
\begin{proofof}{Theorem~\ref{thm:adaptiveResult}}
The proof follows directly by substituting the result of Lemma~\ref{lem:oraclem} into those of Theorem~\ref{thm:adaptiveResult2}, while noting that the second and third term in~\eqref{eq:Noracleresult2} are asymptotically subdominant.
\end{proofof}

As discussed in Ref.\ \cite{WBHS10} for the case of constant time steps, if the Hamiltonian is sufficiently smooth then we can choose $k$ to increase with $\dt$, so that the complexity scales close to linearly in $\Lambdavar\dt$.
Here we obtain a similar result for the case where the time steps are chosen adaptively.
Theorem~\ref{thm:adaptiveResult} provides a guide to choose an optimal value of $k$, which then enables us to prove Corollary~\ref{cor:linscale}.

\begin{proofof}{Corollary~\ref{cor:linscale}}
As~$\{H_\alphavar\}$ is $\Lambda$-$\infty$-smooth, we can choose~$k$ to be any positive integer: in particular we choose~$k=\kz$ where
\begin{equation}
\label{eq:pickkpseudolinear}
	\kz = \left\lceil \sqrt{\frac{1}{2} \log_{25/3}\left(\frac{d^2\bar\interval(\initt+\dt,\initt)\dt}{\epsilon}\right)}~\right\rceil.
\end{equation}
Equation~\eqref{eq:pickkpseudolinear} implies that~$\kz \in \left({d^2\bar\interval(\initt+\dt,\initt)\dt}/{\epsilon}\right)^{o(1)}$ and
$(25/3)^{\kz}\in{(d^2\bar\interval(\initt+\dt,\initt)\dt/\epsilon)^{o(1)}}$
because the square-root of a logarithm grows slower than a logarithm.
We also have
\begin{equation}
	 \left({d^2\bar\interval(\initt+\dt,\initt)\dt}/{\epsilon}\right)^{1/2\kz} \in {(d^2\bar\interval(\initt+\dt,\initt)\dt/\epsilon)^{o(1)}}\label{eq:sublinear},
\end{equation}
because~$\kz$ increases with~$\bar\interval(\initt+\dt,\initt)\dt/\epsilon$.
We then obtain the scaling given in Eq.~\eqref{eq:Noraclescale2} by substituting these expressions into~\eqref{eq:Noracleresult2} and using Lemma \ref{lem:oraclem}.
\end{proofof}

For Corollary~\ref{cor:linscale} it was required that there exists $\scalconst>0$ such that~$\Upsilonfcn'(t)\le \scalconst^2\Upsilonfcn^2(t)$, because this was part of the conditions of Theorem \ref{thm:adaptiveResult}. However, it can be expected that this requirement is automatically satisfied when $\{H_\alphavar\}$ is $\Upsilon$-$\infty$-pointwise-smooth, provided $\Upsilon$ is chosen appropriately (see Appendix~\ref{appendix:derivbd}).

Before concluding, we should also estimate the space-complexity of our algorithm.
It follows from an analysis of the BACS decomposition algorithm, that the number of qubits needed for our simulation is $O(n(\log^*n)^2)$. Unlike the BACS simulation algorithm~\cite{BACS07}, this number does not depend on $\|H\|$ and $\epsilon$.

\section{Conclusions}
We introduce in this paper a pair of quantum algorithms that can be used to simulate time-dependent Hamiltonian evolution on quantum computers, using constant-size or adaptively chosen time steps.
The adaptive time step method can provide superior performance, but requires that upper bounds on the norm of the Hamiltonian and its derivatives are known throughout the simulation.  In both cases, these simulation algorithms can be performed with similar query complexity to the BACS simulation algorithm~\cite{BACS07}
if $H(t)$ is a sum of sufficiently smooth terms.

We also show how to resolve pathological examples, such as our earlier example~\cite{WBHS10} wherein the higher-order derivatives of~$H(t)$ diverge at one point, although this method cannot be used to attain near-linear scaling unless $H(t)$ is a sum of terms that are piecewise sufficiently smooth.
  Furthermore, we have shown that the number of operations used in a simulation of time-dependent
Hamiltonian evolution can be reduced by using lower-order
Lie-Trotter-Suzuki formul\ae~to approximate time-evolution near singularities in the Hamiltonian, and higher-order
formul\ae~farther away from the singularities.  This approach may also be useful in approximating
the solutions to singular differential equations on classical computers.

{It may be difficult to compute some of these quantities, such as~$\Lambdavar$, for some Hamiltonians.  In such circumstances our simulation schemes can
still be used, but the output state of the simulation may not be correct within an error tolerance of~$\epsilon$.  We recommend that heuristic testing be used to estimate whether the
error is within this tolerance in such circumstances.  Such a method could involve performing the swap test~\cite{BCWW01} between the output states of two separate
simulations that employ distinct values of the uncertain parameter.
However values such as~$n$ and
the upper bound for $\|\sparseham_\alphavar\|$ that the oracle uses to encode the matrix elements must be known
in order to perform the simulation.}

We quantify the computational complexity by the number of calls to the oracles~$\fynoarg$, $\BBnoarg$, and~$\{T_\alphavar\}$. These oracles would not typically be fundamental operations, but would be quantum subroutines consisting of sequences of fundamental quantum operations.  The oracles can be
implemented \emph{efficiently} by the quantum computer if each~$\sparseham_{\alphavar}(t)$ is row-computable at every time during the simulation, and there are efficient quantum circuits for the basis transformations $T_\alphavar$.

 This work leaves open several interesting avenues for investigation.  One such avenue is to address the question of whether
 or not strictly linear-time quantum simulation algorithms are possible if the Hamiltonian is time-dependent and sufficiently smooth.
 In addition, it would  be interesting to determine whether or not it is possible to devise an algorithm for which the query complexity scales
 poly-logarithmically with the reciprocal of the error tolerance, rather than sub-polynomially as our algorithm does.

\appendix
\section{Proof of Lemma~\ref{lem:adError}}\label{sec:round-off}
In this section we present a proof of Lemma~\ref{lem:adError}.
It provides values for both $\numQubitsInH$ and~$\numBitsInT$ that ensure that the discretization error is bounded above by $\epsilon/2$.

\begin{proofof}{Lemma~\ref{lem:adError}}
Here we define $\sigma:=\dt/2^{\numBitsInT}$, and define $\tilde H(t)$ to be an approximation to
the Hamiltonian $H(t)$ that uses $\numQubitsInH$ bit
approximations to the matrix elements.  We use $\tilde{H}_{\mu,j}(t)$ to represent
 an approximation to the Hamiltonian $H_{\mu,j}$ that is formed by taking $\numQubitsInH$ bit
approximations to each matrix element of the one-sparse $H_{\mu,j}$.  Similarly, we define $\tilde{\tau}$ to be the closest mesh point to
the time $\tau$ that is still inside the interval~$\mathcal{I}$. In most cases, the nearest mesh point will be at most a distance of~$\sigma/2$ away
from $\tau$, but there are cases where the distance can be up to~$\sigma$.  This occurs when $\tau$ is near the boundary
of one of the subintervals of~$\mathcal I$.  Because the subintervals can have length no shorter than $\sigma$, there will always be a mesh point within
the subinterval, and the distance will not be greater than $\sigma$.  This implies that~$|\tilde \tau -\tau|< \sigma$.
Then, using this notation, our goal in this proof is to show that values of~$\numBitsInT$ and~$\numQubitsInH$ satisfying the inequalities in~\eqref{eq:n'eq0}
guarantee
%the values of~$\numBitsInT$ and~$\numQubitsInH$ that satisfy the inequalities in~\eqref{eq:n'eq0} guarantee
\begin{equation}
\left\|\prod_{p=1}^{N_{\exp}}T_{\mu_p}^{\dagger}\exp[{-iH_{\mu_p,j_p}(\tau_p)\dt_p}]T_{\mu_p}- \prod_{p=1}^{N_{\exp}}T_{\mu_p}^{\dagger}\exp[{-i\tilde H_{\mu_p,j_p}(\tilde{\tau}_p)\dt_p}]T_{\mu_p}\right\|\le\epsilon/2,\label{eq:round-offgoal}
\end{equation}
where~$H_{\mu_p,j_p}$ denotes an element from
a particular sequence of one-sparse Hamiltonians.
%This implies that the total error is $\tilde\epsilon$ because the truncation error in the Lie-Trotter-Suzuki formula has been taken to be $\tilde\epsilon/2$.

Using Eq.~\eqref{eq:mike} we find that, for Hermitian operators $A$ and~$B$,
\begin{align}
\left\| e^{-iA} - e^{-iB}\right\| &= \lim_{n \to \infty} \left\| \exp(-iA/n)^n -\exp(-iB/n)^n\right\| \nn
&\le \lim_{n \to \infty} n \left\| \exp(-iA/n) -\exp(-iB/n)\right\| \nn
&= \lim_{n \to \infty} [\|A-B\| +O(1/n)] = \|A-B\|.
\end{align}
Using this result with the Hermitian operators $H_{\mu_p,j_p}(\tau_p)$ and~$\tilde H_{\mu_p,j_p}(\tilde{\tau}_p)$ gives
\begin{equation}
\|\exp[{-iH_{\mu_p,j_p}(\tau_p)\dt_p}]-\exp[{-i\tilde H_{\mu_p,j_p}(\tilde\tau_p)\dt_p}]\|\leq\|H_{\mu_p,j_p}(\tau_p)-\tilde H_{\mu_p,j_p}(\tilde{\tau}_p)\|\dt_p.\label{eq:errbd1exp}
\end{equation}
Then we obtain, using Eq.~\eqref{eq:mike},
\begin{align}
& \left\|\prod_{p=1}^{N_{\exp}}T_{\mu_p}^{\dagger}\exp[{-iH_{\mu_p,j_p}(\tau_p)\dt_p}]T_{\mu_p}
- \prod_{p=1}^{N_{\exp}}T_{\mu_p}^{\dagger}\exp[{-i\tilde H_{\mu_p,j_p}(\tilde\tau_p)\dt_p}]T_{\mu_p}\right\| \nn
&\leq \sum_{p=1}^{N_{\exp}}\|H_{\mu_p,j_p}(\tau_p)- \tilde H_{\mu_p,j_p}(\tilde\tau_p)\| |\dt_p|
 \le \max_{p} \left\{\|H_{\mu_p,j_p}(\tau_p)- \tilde H_{\mu_p,j_p}(\tilde\tau_p)\|\right\}\sum_{p=1}^{N_{\exp}}|\dt_p|.\label{eq:allbnd}
\end{align}

Our next step is to bound the sum $\sum_{p=1}^{N_{\exp}}|\dt_p|$.  To do so we note that, for the simulation, the time is broken up into~$r$ short intervals, and on each of these the Lie-Trotter-Suzuki formula $U_k$ is used. We denote the $r$ intervals by $\mathcal{I}_1,\mathcal{I}_2,\ldots,\mathcal{I}_r$, and the durations of these intervals by $T_1,T_2,\ldots,T_r$. Using Eq.~(A.3) of Ref.\ \onlinecite{WBHS10}, if the \Th{p} exponential is part of the \Th{q} interval, then the duration of that exponential, $\dt_p$, is at most
\begin{equation}
|\dt_p|\le (2k/3^k)T_q.\label{eq:dtformula}
\end{equation}
The Lie-Trotter-Suzuki formula $U_k$ is composed of~$12Md^2 5^{k-1}$ exponentials, each with a duration that is bounded above by~\eqref{eq:dtformula}.
In addition, $\sum_{q=1}^r T_q\le\dt$, so the total duration of the exponentials used to simulate the evolution in the interval~$\mathcal{I}_q$ is at most
\begin{equation}
\label{eq:sumtbnd}
\sum_{p:\tau_p\in \mathcal{I}_q}|\dt_p|\le 8kMd^2(5/3)^{k-1}T_q
\le 8kMd^2(5/3)^{k-1}\Delta t.
\end{equation}
Lemma~\ref{lem:adError} can be used generally in this work, because this relation does not require that the $r$ intervals have the same duration, or that
$\mathcal{I}$ is a continuous time interval.
The only requirement we have used is that the Lie-Trotter-Suzuki integrator $U_k$ has been used.
This is to obtain the relation \eqref{eq:dtformula} and the number of exponentials in the integrator.

Next, using the triangle inequality, we have
\begin{equation}
\label{eq:htri}
	\left\|H_{\mu_p,j_p}(\tau_p)- \tilde H_{\mu_p,j_p}(\tilde\tau_p)\right\|
		\le \left\|H_{\mu_p,j_p}(\tau_p)- H_{\mu_p,j_p}(\tilde{\tau}_p)\right\|
			+\left\|H_{\mu_p,j_p}(\tilde\tau_p)- \tilde H_{\mu_p,j_p}(\tilde\tau_p)\right\|.
\end{equation}
Using Taylor's theorem, we find that an upper bound for the error due to the time discretization is
\begin{equation}
	\left\|H_{\mu_p,j_p}(\tau_p)-H_{\mu_p,j_p}(\tilde\tau_p) \right\|
		\le \max_{t,\mu}\left\|\partial_t H_{\mu}(t)\right\|\sigma.\label{eq:HDerivBd}
\end{equation}
By using a value of~$\numBitsInT$ that satisfies~\eqref{eq:n'eq0}, we obtain
\begin{equation}
	 \left\|H_{\mu_p,j_p}(\tau_p)-H_{\mu_p,j_p}(\tilde\tau_p)\right\|\le\frac{\epsilon}{(32kMd^2)(5/3)^{k-1}\dt}.
\end{equation}

Next we consider the error due to the discretization of~$\sparseham$. The matrix elements are encoded in polar form, so errors emerge because of inaccuracies in the modulus as well as the phase. Using the triangle inequality, the error is bounded above by $\epsilon_{\rho}+\epsilon_{\phi}\upp$, where~$\epsilon_{\phi}$ is the discretization error in the phase, $\epsilon_{\rho}$ is the error in the modulus, and~$\upp$ is an upper bound for the magnitudes of the matrix elements. We choose the same number of bits to encode the modulus and the phase.  Then taking $\numQubitsInH$ to exceed the value in~\eqref{eq:n'eq0}, the error in the modulus and phase satisfy
\begin{align}
\epsilon_{\rho} &\le \upp/ \left(2^{\numQubitsInH/2}\right) \le \frac 18 \frac{\epsilon}{(32kMd^2)(5/3)^{k-1}\dt} \\
\epsilon_{\phi} &\le 2\pi/ \left(2^{\numQubitsInH/2}\right) \le \frac 18 \frac{2\pi\epsilon}{(32kMd^2)(5/3)^{k-1}\upp \dt}
\end{align}
Using these relations, we obtain
\begin{equation}
	\left\|H_{\mu_p,j_p}(\tilde\tau_p)- \tilde H_{\mu_p,j_p}(\tilde\tau_p)\right\|
		\le \epsilon_{\rho}+\epsilon_{\phi}\upp \le \frac{\epsilon}{(32kMd^2)(5/3)^{k-1}\dt}.
\end{equation}
Inserting Eqs.\ \eqref{eq:sumtbnd} and \eqref{eq:htri} into Eq.~\eqref{eq:allbnd} gives the error bound as
\begin{align}
	\max_{p} \biggr\{\biggr\|H_{\mu_p,j_p}(\tau_p)-&\tilde{H}_{\mu_p,j_p}(\tilde\tau_p)\biggr\|\biggr\}\sum_{p=1}^{N_{\exp}}|\dt_p|\le \frac{2\epsilon}{(32kMd^2)(5/3)^{k-1}\dt}\times 8kMd^2(5/3)^{k-1}\Delta t = \epsilon/2.
\end{align}
\end{proofof}

\section{Bounds On Derivatives of~$\Upsilonfcn$ \label{appendix:derivbd}}
In our adaptive simulation method we have required that~$\Upsilon$ is chosen such that there exists a constant $\scalconst$ such that~$|\partial_t\Upsilonfcn(t)| \le \scalconst^2[\Upsilonfcn(t)]^2$ for all times in the interval.  We show in this appendix that this requirement is natural
by demonstrating that it naturally emerges for Hamiltonians where~$\Upsilon$ is chosen to be the smallest permissible function.

Now define the set of functions $\{\Upsilon_P\}$ to be the smallest possible functions such that~$\{H_\alphavar\}$ is $\Upsilon_P$-$P$-pointwise-smooth.
Taking the derivative of~$\Upsilonfcn_{P}(t)$ gives
\begin{align}
	\Upsilonfcn_{P}'(t)
		&= \lim_{\delta t\to 0} \frac{\Upsilonfcn_{P}(t+\delta t)-\Upsilonfcn_{P}(t)}{\delta t} \nonumber \\
		&\le \sup\left\{\left.\left(\partial_{u} \left\|H(u)\right\|\right)\right|_{u=t},\cdots,
		\left.\left(\partial_{u} \|\partial_{u}^{p-1}H(u)\|^{1/p}\right)\right|_{u=t},\cdots\,\left.\left(\partial_{u}
		 \left\|\partial_{u}^{P}H(u)\right\|^{1/(P+1)}\right)\right|_{u=t}\right\} \nonumber \\
		&\le \sup\left\{\frac 1p
		 \left\|H^{(p-1)}(t)\right\|^{1/p-1}\left\|H^{(p)}(t)\right\|:p=1,\ldots,P+1\right\}.
\end{align}
There are now two possible cases, either
\begin{equation*}
	 \left\|H^{(p-1)}(u)\right\|^{1/p}\leq\left\|H^{(p)}(u)\right\|^{1/(p+1)}
\end{equation*}
or
\begin{equation*}
	 \left\|H^{(p-1)}(u)\right\|^{1/p}\geq\left\|H^{(p)}(u)\right\|^{1/(p+1)}.
\end{equation*}
In the first case we obtain
\begin{equation}
	\frac 1p\left\|H^{(p-1)}(t)\right\|^{1/p-1}\left\|H^{(p)}(t)\right\| \le \frac 1p\left\|H^{(p)}(t)\right\|^{2/(p+1)},
\end{equation}
and in the second case
\begin{equation}
	\frac 1p\left\|H^{(p-1)}(t)\right\|^{1/p-1}\left\|H^{(p)}(t)\right\| \le \frac 1p\left\|H^{(p-1)}(t)\right\|^{2/p}.
\end{equation}
From the definition of~$\Upsilonfcn_{P}(t)$, $\|H^{(p)}(t)\|^{1/(p+1)}\leq \Upsilonfcn_{P}(t)$ for $p=0,1,...,P$.
However, because~$\|H^{(P+1)}(t)\|^{1/(P+2)}$ is not included in the definition of~$\Upsilonfcn_{P}(t)$, we use $\|H^{(P+1)}(t)\|^{1/(P+2)}\leq \Upsilonfcn_{P+1}(t)$ to bound it.
Then, using the fact that~$\Upsilonfcn_{P+1}(t)\geq \Upsilonfcn_{P}(t)$ and~$p\ge 1$, the derivative of~$\Upsilonfcn_{P}(t)$ is bounded above by
\begin{equation}
\Upsilonfcn_{P}'(t) \le [\Upsilonfcn_{P+1}(t)]^2.
\end{equation}
A lower bound for the derivative can be obtained in the same way, giving the general result
\begin{equation}
|\Upsilonfcn_{P}'(t)| \le [\Upsilonfcn_{P+1}(t)]^2.\label{eq:upsilonPderiv}
\end{equation}
If there exists a constant $\scalconst$ such that for all $t\in[\initt,\initt+\dt]$, $\Upsilonfcn_{P+1}(t)\le \scalconst\Upsilonfcn_{P}(t)$, then we obtain the restriction in Theorem \ref{thm:adaptiveResult}, $|\Upsilonfcn'_{P}(t)| \le \scalconst^2[\Upsilonfcn_{P}(t)]^2$.

In the case where~$\Upsilonfcn_\infty$ is taken to be the smallest possible function such that~$\{H_\alphavar\}$ is $\Upsilonfcn_\infty$-$\infty$-pointwise-smooth, this restriction need not be made.  We see from taking the limit as~$P\to\infty$ of~\eqref{eq:upsilonPderiv} that
\begin{equation}
	|\Upsilonfcn_{\infty}'(t)| \le [\Upsilonfcn_{\infty}(t)]^2.
\end{equation}
This means that, if $\{H_\alphavar\}$ is $\Upsilon$-$\infty$-pointwise-smooth, the condition~$|\Upsilonfcn'(t)| \le [\Upsilonfcn(t)]^2$ should hold if $\Upsilonfcn(t)$ is chosen appropriately.
(It does not imply this condition, because~$\Upsilonfcn(t)$ could be chosen poorly.)

\acknowledgements
NW thanks A. Hentschel and A. Childs for many helpful comments.  We acknowledge MITACS research network,
General Dynamics Canada, USARO and \textit{i}CORE for financial support.
PH is a CIFAR Scholar, and BCS is a CIFAR Fellow.


\begin{thebibliography}{10}
\bibitem{endnote}
There is a typographical error in Definition 4 of Ref.\ \cite{WBHS10}. The definition should also allow $p=0$.
\end{thebibliography}
\end{document}